\PassOptionsToPackage{unicode}{hyperref}
\PassOptionsToPackage{hyphens}{url}
\PassOptionsToPackage{dvipsnames,svgnames,x11names}{xcolor}
\documentclass[12pt,fleqn]{article}

\usepackage{amsmath,amssymb}
\usepackage{bm}
\usepackage{iftex}
\usepackage{float}
\usepackage{ulem}
\usepackage{amsthm}

\usepackage{textcomp}

\usepackage{xr}
\makeatletter
\newcommand*{\addFileDependency}[1]{
  \typeout{(#1)}
  \@addtofilelist{#1}
  \IfFileExists{#1}{}{\typeout{No file #1.}}
}
\makeatother

\newcommand*{\myexternaldocument}[1]{%
    \externaldocument{#1}%
    \addFileDependency{#1.tex}%
    \addFileDependency{#1.aux}%
}

\myexternaldocument{SM}
\usepackage[ruled,vlined,linesnumbered]{algorithm2e}

\SetCommentSty{mycommfont}

\newtheorem{proposition}{Proposition}

\usepackage{lmodern}
\ifPDFTeX\else  
\fi
\IfFileExists{upquote.sty}{\usepackage{upquote}}{}
\IfFileExists{microtype.sty}{
  \usepackage[]{microtype}
  \UseMicrotypeSet[protrusion]{basicmath} 
}{}
\makeatletter
\@ifundefined{KOMAClassName}{
  \IfFileExists{parskip.sty}{%
    \usepackage{parskip}
  }{
    \setlength{\parindent}{0pt}
    \setlength{\parskip}{6pt plus 2pt minus 1pt}}
}{
  \KOMAoptions{parskip=half}}
\makeatother
\usepackage{xcolor}
\makeatletter
\ifx\paragraph\undefined\else
  \let\oldparagraph\paragraph
  \renewcommand{\paragraph}{
    \@ifstar
      \xxxParagraphStar
      \xxxParagraphNoStar
  }
  \newcommand{\xxxParagraphStar}[1]{\oldparagraph*{#1}\mbox{}}
  \newcommand{\xxxParagraphNoStar}[1]{\oldparagraph{#1}\mbox{}}
\fi
\ifx\subparagraph\undefined\else
  \let\oldsubparagraph\subparagraph
  \renewcommand{\subparagraph}{
    \@ifstar
      \xxxSubParagraphStar
      \xxxSubParagraphNoStar
  }
  \newcommand{\xxxSubParagraphStar}[1]{\oldsubparagraph*{#1}\mbox{}}
  \newcommand{\xxxSubParagraphNoStar}[1]{\oldsubparagraph{#1}\mbox{}}
\fi
\makeatother

\usepackage{longtable,booktabs,array}
\usepackage{calc} 
\usepackage{etoolbox}
\makeatletter
\patchcmd\longtable{\par}{\if@noskipsec\mbox{}\fi\par}{}{}
\makeatother
\IfFileExists{footnotehyper.sty}{\usepackage{footnotehyper}}{\usepackage{footnote}}
\makesavenoteenv{longtable}
\usepackage{graphicx}
\makeatletter
\def\maxwidth{\ifdim\Gin@nat@width>\linewidth\linewidth\else\Gin@nat@width\fi}
\def\maxheight{\ifdim\Gin@nat@height>\textheight\textheight\else\Gin@nat@height\fi}
\makeatother
\setkeys{Gin}{width=\maxwidth,height=\maxheight,keepaspectratio}
\makeatletter
\def\fps@figure{htbp}
\makeatother

\makeatletter
\@ifpackageloaded{caption}{}{\usepackage{caption}}
\AtBeginDocument{%
\ifdefined\contentsname
  \renewcommand*\contentsname{Table of contents}
\else
  \newcommand\contentsname{Table of contents}
\fi
\ifdefined\listfigurename
  \renewcommand*\listfigurename{List of Figures}
\else
  \newcommand\listfigurename{List of Figures}
\fi
\ifdefined\listtablename
  \renewcommand*\listtablename{List of Tables}
\else
  \newcommand\listtablename{List of Tables}
\fi
\ifdefined\figurename
  \renewcommand*\figurename{Figure}
\else
  \newcommand\figurename{Figure}
\fi
\ifdefined\tablename
  \renewcommand*\tablename{Table}
\else
  \newcommand\tablename{Table}
\fi
}
\@ifpackageloaded{float}{}{\usepackage{float}}
\floatstyle{ruled}
\@ifundefined{c@chapter}{\newfloat{codelisting}{h}{lop}}{\newfloat{codelisting}{h}{lop}[chapter]}
\floatname{codelisting}{Listing}

\makeatother
\makeatletter
\@ifpackageloaded{caption}{}{\usepackage{caption}}
\@ifpackageloaded{subcaption}{}{\usepackage{subcaption}}
\makeatother

\ifLuaTeX
  \usepackage{selnolig}  
\fi
\usepackage[]{natbib}
\usepackage{cancel}
\usepackage{bookmark}
\usepackage{titlesec}

\titleformat{\subsection}
  {\large\bfseries\itshape}
  {\thesubsection}
  {1em}
  {}
\titleformat{\subsubsection}
  {\itshape}
  {\thesubsubsection}
  {1em}
  {}

\IfFileExists{xurl.sty}{\usepackage{xurl}}{} 
\hypersetup{
  pdftitle={Title},
  pdfauthor={Author 1; Author 2},
  pdfkeywords={3 to 6 keywords, that do not appear in the title},
  colorlinks=true,
  linkcolor={blue},
  filecolor={Maroon},
  citecolor={Blue},
  urlcolor={Blue},
  pdfcreator={LaTeX via pandoc}}

\newcommand{\anon}{1}

\usepackage{setspace}

\begin{document}

\onehalfspacing 


\if1\anon
{
  \title{\bf Parameter estimation in Conditional Sequential Monte Carlo algorithms through Particle Learning}
  \author{Alfonso Diz-Lois Palomares\hspace{.2cm}\\
    Department of Mathematics, University of Oslo\\
    and \\
    Geir Storvik \\
    Department of Mathematics, University of Oslo}
  \maketitle
} \fi

\if0\anon
{
  \bigskip
  \bigskip
  \bigskip
  \begin{center}
    {\LARGE\bf Title}
\end{center}
  \medskip
} \fi

\bigskip
\begin{abstract}
In this work, we explore particle learning strategies for the joint estimation of static parameters and latent states within conditional sequential Monte Carlo (CSMC) algorithms. Building on this idea, we propose the p(parameter)-CSMC algorithm, which incorporates both parameter learning and ancestor sampling,  leading to much better mixing properties compared to (particle) Gibbs sampling in settings where strong internal correlations may challenge effective exploration. We also include two applications in the context of a branching process model: one using synthetic data, where we estimate the infectivity profile while assuming the reproductive number to be known, and another using real data, where we address the joint inference of the reproductive number and the infectivity profile based on daily hospital incidence from the arrival of the SARS-CoV-2 lineage B.1.1.7 (Alpha) in Norway in February 2021. We show that, in these settings, performance is dramatically enhanced, with substantially faster mixing and markedly reduced autocorrelation compared with standard particle Gibbs.
\end{abstract}

\noindent%
\vfill

\newpage

\section{Introduction}\label{sec-intro}

State Space Models (SSM) provide a statistical framework to model dynamic systems that are partially observed over time. They have found widespread applications across multiple fields including Physics \citep{ioanna2026applications}, Economics \citep{durbin2012time}, Ecology \citep{auger2021guide} and Epidemiology \citep{birrell2018evidence,storvik2023sequential} just to name a few.

 In its most basic form, these models are characterized by two main components: the latent space $x_t$ which represents the underlying, unobserved, process of the system at time $t$, and the observations, denoted by $y_t$, which are the measurable outputs influenced by the latent processes. Typically, the latent process follows a Markov structure, while the observations are independent, conditionally on the latent process.
 
In situations where inference about the latent process  is of interest, Sequential Monte Carlo (SMC) methods, commonly known as particle filters (PF), provide a robust solution for performing online inference~\citep[see][for a general introduction]{doucet2001introduction}. For high-dimensional latent variables, these methods can be challenging to apply, however recent approaches shows promise~\citep{malory2021bayesian,finke2023conditional,corenflos2024particle}. For the case when Kalman filter dynamics are assumed in a subset of the state space, different strategies can be utilized to exploit the internal structure, particularly via marginalization/Rao–Blackwellization, both with and without unknown parameters \citep{murphy2001rao,schon2005marginalized,lindsten2015rao,kok2024rao}. 

The presence of unknown static parameters $\theta$, imposes an additional challenge. In such a setting, different approaches are needed to perform inference in the combined space (for a comprehensive review, see \cite{kantas2015particle, luengo2020survey}). 

A widely used strategy is to treat static parameters as part of the state and introduce artificial dynamics \citep{kitagawa1998self,liu2001combined}. This improves mixing and alleviates particle degeneracy because the particles can explore parameter space and resampling remains effective. However, this comes at the cost of turning the algorithm into an approximation.

Nested SMC methods, like the $\text{SMC}^2$ \citep{chopin2013smc2}, place an outer SMC in the parameter dimension while each parameter particle carries an inner particle filter to estimate the likelihood. This enables sequential Bayesian learning of the posterior distribution of the parameters $p(\theta|y_{1:t})$ with resample–move steps for rejuvenation, albeit at higher computational cost. 

The PaRIS (Particle-based Rapid Incremental Smoother) algorithm \citep{westerborn2014efficient,cardoso2023state} is a complementary alternative for handling unknown parameters and mitigating particle path degeneracy in certain state-space models. PaRIS addresses degeneracy through smoothing of additive functionals, and is particularly well-suited for non-Bayesian estimation schemes such as maximum likelihood and online Expectation-Maximization (EM). 

When conjugacy is available, Particle Learning \citep{storvik2002particle,fearnhead2002markov,Carvalho2010PLS} is particularly attractive as it propagates sufficient statistics and updates parameter posteriors analytically within the particle system, delivering fast unbiased updates and often superior scalability in practice. As pointed out in \citet{andrieu2005line}, the problem with this approach is that the SMC estimates of
the sufficient statistics necessary to perform the parameter updates
degrade as the time horizon $T$ increases because they are based on the particle approximation of the joint posterior distribution for all the latent variables $x_1,...,x_T$, which degenerates with time. In some cases, marginalization of the parameters is possible, but at the cost of introducing more complex non-Markovian dependence structures within the remaining latent variables. This can lead to a high computational burden~\citep{wigren2019parameter}.

As an offline alternative, Particle MCMC (PMCMC) methods use particle filters within MCMC to sample iteratively from the joint posterior of states and parameters using the full dataset \citep{andrieu2010particle}. Two principal PMCMC schemes are commonly employed: Particle Marginal Metropolis–Hastings (PMMH), which embeds a particle-filter used to get an unbiased estimate of the marginal likelihood within a Metropolis–Hastings update for $\theta$; and Particle Gibbs  (PG), which leverages Conditional SMC (CSMC) to construct a Gibbs sampler alternating between latent states and parameters. Because PMMH relies on a Metropolis–Hastings acceptance step and therefore on well-tuned proposals for the parameters, PG is typically preferred whenever the full conditional of the parameters is available. However, when there is strong posterior dependence between parameters and latent states or among parameter components, Gibbs updates may mix poorly, leading to slow exploration of the posterior and highly autocorrelated samples \citep{robert2004monte}.

The Markov transition kernels within CSMC remain “well-behaved” as $T$ grows, under conditions where standard independent SMC proposals tend to suffer from degeneracy of the path and poor acceptance rates \citep{chopin2015particle}.  Due to its offline nature, its major limitation is the need for at least $\mathcal{O}(NT)$ operations per iteration, where $N$ is the number of particles (samples) used in the Monte Carlo approximation. This makes efficient implementations and strong mixing essential. Rejuvenation techniques such as backward sampling \citep{whiteley2010discussion} or ancestor sampling \citep{lindsten2014particle} improve mixing and reduce autocorrelation, changing the error’s dependence on $T$ from linear to constant. However, for non‑Markovian models these methods become costly because weight calculation raises the complexity to $\mathcal{O}(NT^2)$ \citep{lindsten2014particle}.

Recently, \citet{corenflos2025particle} presented an algorithm (the marginal Particle Gibbs (m-PG)) where the unknown parameter can be effectively marginalized out in some augmented target distribution that evaluates multiple proposals at once. These proposals are drawn conditionally independently through an auxiliary variable outside the CSMC step, and not jointly and dynamically intertwined with the latent states, as we propose here.

In this work, we present a CSMC-based framework for simultaneous inference of the latent space and the static parameters in models for which the conditional distribution of the parameters given the latent space is available. We propose a reformulation of the problem by extending the static parameter $\theta$ to a sequence of parameters $\theta_{1:T}$ with $\theta=\theta_T$ and a corresponding augmented distribution keeping the original posterior distribution as a marginal of the extended distribution. An effective CSMC-based algorithm working on this extended space is then proposed. 

We demonstrate the method's applicability when sampling the parameter $\theta$ conditional on the latent variables (and observations) is possible through some sufficient statistic. The use of sufficient statistics improves efficiency and reduces computational cost. This saving is especially important when backward or ancestor sampling is needed to mitigate degeneracy in the time dimension. We verify through experiments how our approach outperforms PGAS in settings with high internal correlations involved where the exploration can be challenging. 

The structure of the paper is as follows. In Section \ref{sec-back}, we provide the necessary background by first formulating the problem and then describing the CSMC algorithm and particle Gibbs strategies for static parameter estimation. In Section \ref{sec-meth}, we define a general framework for implementing particle learning within CSMC and introduce the p-CSMC algorithm with ancestor sampling. In Section \ref{sec-verify}, we present two applications: one using synthetic data, where we estimate the weights in a branching process model, and a more complex example where we estimate both the infectivity profile and the reproductive number during the arrival of the Alpha SARS-CoV-2 variant to Norway. We assess the performance and compare it against standard particle Gibbs. Finally, in Section \ref{sec-conc}, we discuss the results, outline the contexts in which the proposed methodology is particularly relevant, highlight its main strengths and limitations, and indicate directions for future work.

\section{Background}\label{sec-back}

\subsection*{State space models}

Let us consider some latent variables $\{x_t\}$, where each $x_t \in \mathbb{R}^d$, evolving over time through some transition probabilities/dynamic model that might depend on some parameter $\theta$. Additionally, the observations $\{y_t\}$ are connected to the latent space through some potentials/observation process:
\begin{align*}
x_t& \sim p(\cdot|\bm{x}_{0:t-1};\theta);\\
y_t&\sim p(\cdot|x_t) , \quad   
\end{align*}
for $t=1,...,T$,
where we have used $p(\cdot|\cdot)$ generically for distributions involved in the assumed model. 

We here assume $x_0$ is a known quantity.
For the sake of notational simplicity, we adopt the convention of using lowercase letters to denote random variables and do not distinguish in the notation between random variables and their realizations.
Similarly, all $x_t$, $y_t$ and $\theta$ may be vectors, but we will reserve the  boldface notation for sequences of variables $\bm{x}_{s:t}= (x_s,...,x_t)$. We assume that any unknown parameters enter the model through the transition probabilities rather than the observational model. Alternatively, it is often possible under standard conditions to reformulate the model so that the observation distribution no longer depends explicitly on those parameters. To denote this dependency, we use $p_\theta$ and $p(|\theta)$ interchangeably 

Under a first-order Markov structure where $p(x_t|\bm{x}_{1:t-1},\theta) = p(x_t|x_{t-1},\theta)$, the joint distribution of the latent variables and the observations can be factorized into the familiar sequence:
\begin{equation}
p_{\theta}(\bm{x}_{1:T}|\bm{y}_{1:T}) \propto p_{\theta}(\bm{x}_{1:T},\bm{y}_{1:T}) = 
\prod_{t=1}^T p(x_t|{x}_{t-1},\theta) 
p(y_t|x_{t}).  
\label{eq:posteriorsmc}
\end{equation}

In the presence of unknown static parameters, we consider a full Bayesian approach, including a prior $p(\theta)$ for $\theta$. In this case, 
\begin{align}
p(\bm{x}_{1:T},\theta|\bm{y}_{1:T}) \propto& p(\theta)
\prod_{t=1}^T p(x_t|x_{t-1},\theta) 
p(y_t|x_{t}) 
\label{eq:target_param}
\intertext{which also can be rewritten to}
p(\bm{x}_{1:T},\theta|\bm{y}_{1:T}) \propto& 
p(\theta|\bm x_{1:T})
\prod_{t=1}^T p(x_t|x_{1:t-1}) 
p(y_t|x_{t})
\end{align}
showing that the Markov structure in \eqref{eq:posteriorsmc} is lost.

\subsection*{The SMC/CSMC framework}

A Feynman-Kac model is defined as
\begin{align}
\mathbb{Q}_{T}(\bm{z}_{1:T})=\frac{1}{L_T} M_1(z_1)G_1(z_1)\prod_{s=2}^TM_s(z_{1:s-1},z_s)G_s(z_s)
\label{eq:simple_target}
\end{align}
where $L_T$ is the normalising constant needed for $\mathbb{Q}_T$ to be a probability measure \citep{chopin2020feynman}. In this formulation, we use a more generic $z_t$ variable, which in some cases will correspond to $x_t$ but, as we will see later, can also include other variables.

Assuming all transition kernels $M_t$ admit a density with respect to Lebesgue measure on $\mathbb{R}^d$, we can then define a sequence of unnormalized joint probabilities:
\begin{equation}
\gamma_t(\bm z_{1:t})= M_1(z_1)G_1(z_1)\prod_{s=2}^tM_s(z_{1:s-1},z_s)G_s(z_s)
\label{eq:gammas}
\end{equation}
where $\mathbb{Q}_T(\bm z_{1:T})\propto \gamma_T(\bm z_{1:T})$.
The following recursion applies:
\begin{align*}
    \gamma_t(\bm z_{1:t})=\gamma_{t-1}(\bm z_{1:t-1})M_t(\bm z_{1:t-1},z_t)G_t(z_t).
\end{align*}
As long as $\gamma_T$ is the main target, there is some flexibility in the choice of $M_s$ and $G_s$, see e.g.~\citet{guarniero2017iterated}. 

Model~\eqref{eq:posteriorsmc}, with $\theta$ known, is a special case of~\eqref{eq:simple_target} using $z_t=x_t$, $M_s(z_{1:s-1},z_s)=p(x_s|x_{s-1};\theta)$ and $G_s(z_s)=p(y_s|x_s)$.  We then have
\begin{align*}
p_\theta(\bm x_{1:t}|\bm y_{1:t})\propto \gamma_t(\bm x_{1:t})=&p_\theta(\bm x_{1:t-1}|\bm y_{1:t-1})p(x_t|\bm{x}_{t-1};\theta)p(y_t|x_t).
\end{align*}
In the case with unknown static parameters involved, we may consider the marginal distributions $M_s(\bm x_{1:t-1},x_t)=p(x_t|\bm x_{1:t-1})$, giving
\begin{align*}
 \gamma_t(\bm x_{1:t})=&p(\bm x_{1:t-1}|\bm y_{1:t-1})p(x_t|\bm{x}_{1:t-1})p(y_t|x_t)\propto p(\bm x_{1:t}|\bm y_{1:t})
\end{align*}
where now $p(x_t|\bm x_{1:t-1})=\int_{\theta}p(x_t|x_{t-1},\theta)p(\theta|\bm x_{1:t-1})d\theta$. Note that in this case, the $M_s$ distributions have a non-Markovian structure. 

Sampling can be performed through sequential Monte Carlo (SMC).
Assuming $\bm z_{1:t-1}^{(i)}$ is drawn from (the normalized version of) $\gamma_{t-1}(\bm z_{1:t-1})$ and $z_t^{(i)}$
from some proposal distribution $q(z_t|z_{1:t-1}^{(i)})$, importance weights for updating $\bm z_t^{(i)}=(\bm z_{t-1}^{(i)},z_t^{(i)})$ to a sample from $\gamma_t(\bm z_t)$
is given by
\begin{align}
w_t^{(i)} &= 
\frac{\gamma_t(\bm z^{(i)}_{1:t})}
     {\gamma_{t-1}(\bm z^{(i)}_{1:t-1})\,q_t(z^{(i)}_t|\bm z^{(i)}_{1:t-1})}
=
\frac{M_t(\bm z^{(i)}_{1:t-1},z^{(i)}_t)\,G_t(z^{(i)}_t)}
     {q_t(z^{(i)}_t|\bm z^{(i)}_{1:t-1})},\label{eq:w_z}
\end{align}
and the normalized equivalent
$W_t^{(i)} = w_t^{(i)}/\sum_{j=1}^N w_t^{(j)}$
when samples for $i=1,...,N$ are generated. If we use the transition kernel of the model as our proposal,
$q_t(z_t|\bm z_{1:t-1})=M_{t}(\bm z_{1:t-1},z_t)$, we get
$w_t^{(i)}=G_t(z_t^{(i)})$,
which corresponds to the bootstrap filter~\citep{gordon1993novel}. 



In non-Markovian models simulating from $M_t$ might be difficult. For the specific setting where $M_t$ corresponds to $p(x_t|\bm x_{1:t-1})$, in the presence of an unknown parameter $\theta$, simulation can be performed through the two-step procedure
 \begin{enumerate}
\item Simulate $\theta\sim p(\theta|\bm x_{1:t-1})$;
\item Simulate $x_t\sim p(x_t|\theta,\bm x_{1:t-1})$,
 \end{enumerate}
 which is exactly what is used in the sufficient statistic approach \citep{storvik2002particle,fearnhead2002markov}. Although ordinary SMC algorithms have several desired properties, including consistency as the number of particles $N$ increases \citep{DelMoral2004}, with large $T$ this procedure may suffer from serious degeneracy problems, in particular for the first time points~\citep{andrieu2005line}.

Another possibility is to adopt an offline approach  and target the joint distribution $p(\bm x_{1:T},\theta|y_{1:T})$ through Particle Gibbs \citep{andrieu2010particle}. The idea is to iteratively draw samples of the parameters from $p(\theta|\bm x_{1:T}^*)$, and then simulating a new path $\bm x_{1:T}$ through a CSMC sampler, with $\bm x_{1:T}^*$ as the reference path. The CSMC step is invariant with respect to the distribution $p(\bm x_{1:T}|\theta,\bm y_{1:T})$.
Algorithm~\ref{alg:PG_AS} describes such a procedure, utilizing the generic CSMC algorithm~\ref{alg:CSMC_AS}, and it will constitute the reference for comparison against the other approaches presented in the following sections. In practice, the only modification relative to ordinary SMC, except for the ancestor sampling step, is the presence of a reference trajectory $\bm z_{1:T}^{*}$ as an additional particle that is allowed to survive all the resampling steps throughout the iterations over $t$.

In this MCMC setting, the relevant notion of accuracy is the convergence of the overall Markov chain, which depends on the number of MCMC iterations. On the other hand, as the number of particles $N$ in the CSMC increases, the resulting Particle Gibbs kernel approaches that of an ideal Gibbs sampler in terms of mixing \citep{andrieu2010particle}. 

\subsection*{Backward/Ancestor sampling}

The vanilla CSMC kernel suffers from severe path degeneracy as the time horizon $T$ grows (similar to the sufficient statistics approach). This leads to highly correlated updates of $\bm z_{1:T}$ across PG iterations and poor mixing, especially for the initial states. A way to mitigate this issue is to incorporate backward-sampling moves that rejuvenate the ancestry of the reference path. Conceptually, these moves resample the ancestors of the current reference at time $t$  ($\bm z^*_{t:T}$) given the particles available at time $t-1$, thereby restoring diversity in the early part of the path while preserving the CSMC invariant distribution. One can implement this either via a backward simulation sweep after the forward pass~\citep{whiteley2010discussion} or, more simply, by performing on-the-fly backward resampling during the forward pass~\citep[ancestor sampling,][]{lindsten2014particle} which is the approach considered here.

In its general formulation, ancestor sampling 
allows the reference path to update ($\bm z_{1:t-1}^*$)  during the forward iteration by sampling one of the particles using the ancestor weights:
\begin{align}
\label{eq:weightsas}
\widetilde{w}_{t-1,T}^{(i)} \triangleq& 
w_{t-1}^{(i)}\frac{\gamma_T(\langle \bm z_{1:t-1}^{(i)},\bm z_{t:T}^*\rangle)}{\gamma_t(\bm z_{1:t-1}^{(i)})}
=\prod_{s=t}^TM_s(\langle\bm z_{1:t-1}^{(i)},\bm z_{t:s-1}^*\rangle,z_s^*)G_s(z_s^*).
\end{align}
where $\langle \bm a,\bm b\rangle$ is the concatenation of vectors $\bm a$ and $\bm b$.
The forward weight $w_{t-1}^{(i)}$ corrects for the prior probability of $\bm z^{(i)}_{t-1}$ and the ratio of the target densities can be seen as the likelihood that $\bm z_{t:T}^{*}$ originated from $\bm z^{(i)}_{1:t-1}$ \citep{lindsten2014particle},. 

In non-Markovian models, equation \eqref{eq:weightsas} leads to ancestor weights whose computational cost scales as $\mathcal{O}(NT)$ for each time-point~\citep{lindsten2014particle}. However, in those cases where the Markov property applies (e.g. there are no unknown parameters or they are fixed as in PG) the expression simplifies further to 
\[
\widetilde{w}_{t-1,T}^{(i)} \propto w_{t-1}^{(i)} M_t(z_{t-1}^{(i)},z_t),
\]
due to that the remaining terms do not depend on $i$ resulting in a computational cost that now scales to $\mathcal{O}(N)$.


\begin{algorithm}[t]
\DontPrintSemicolon
\caption{CSMC with Ancestor Sampling (CSMC-AS)}
\label{alg:CSMC_AS}
\KwIn{Reference trajectory $\bm{z}^*_{1:T}$}
\KwOut{New sample $\bm z_{1:T}$ leaving $Q_T(\bm{z}_{1:T})$ as invariant density}
Draw $z_1^{(i)} \sim q_1(\cdot)$ for $i=1,\ldots,N-1$ and
set $z_1^{(N)}=z^*_1$ \\
Compute $w_1^{(i)}$ according to~\eqref{eq:w_z} and $\widetilde w_1^{(i)}$ according to~\eqref{eq:weightsas} for $i=1,\cdots,N$ \\
\For{$t=2$ \KwTo $T$}{

Draw ancestors $\{a_t^i\}$ with $\mathbb{P}(a^i_t = k) \propto w^{(k)}_{t-1} \text{ for } i = 1, \ldots, N-1$ \\
Draw $a_t^N$ with $\mathbb{P}(a_t^N=i)\propto \widetilde{w}^{(i)}_{t-1|T}$\\
Draw $z_t^{(i)} \sim q_t(\cdot |z^{(a_t^i)}_{t-1})$ for $i=1,\ldots,N-1$ and set $z_t^{(N)}=z^*_t$ \\
Set $\bm{z}_{1:t}^{(i)}=\{\bm{z}_{1:t-1}^{(a_t^i)},z_t^{(i)}\}$ for $i=1,\ldots,N$\\
Compute $w_t^{(i)}$ according to \eqref{eq:w_z} and $\widetilde w_t^{(i)}$ according to \eqref{eq:weightsas} for $i=1,\ldots,N$
}
Draw $k$ with $\mathbb{P}(k=i)\propto w_T^{(i)}$
\;
\textbf{return} $\bm{z^{(k)}}_{1:T}$
\end{algorithm}

\begin{algorithm}[t]
\DontPrintSemicolon
\caption{PG-AS}
\label{alg:PG_AS}
\KwIn{Reference trajectory $\bm{x}^*_{1:T}$, reference parameters $\theta^*$}
\KwOut{New sample $(\bm x_{1:T},\theta)$ leaving $p(\bm x_{1:T},\theta|\bm y_{1:T})$ as invariant density}

\tcc{State update: CSMC-AS}
Sample $\bm x_{1:T}$ through CSMC-AS (Algorithm \ref{alg:CSMC_AS}) with $z_t=x_t,M_t(\bm z_{1:t-1},z_t)=p(x_t|x_{t-1},\theta^*)$ and $G_t(z_t)=p(y_t|x_t)$ 
 \\
\tcc{Parameter update}
Sample $\theta \sim p(\theta \mid \bm{x}_{1:T}, \bm{y}_{1:T})$
\;
\BlankLine
\textbf{return} $(\bm{x}_{1:T}, \theta)$
\end{algorithm}

\section{Methods}\label{sec-meth}

In this section, we propose a family of CSMC algorithms for simultaneous updates of the latent process $\bm x_{1:T}$ and the parameter $\theta$. 
Our main target is the joint distribution
\begin{align}
p(\bm x_{1:T},\theta|\bm y_{1:T})\propto &p(\theta)
\left[
\prod_{s=1}^Tp(x_s|x_{s-1},\theta)p(y_s|x_s)\right]
\label{eq:target_smc_nonMarkovian}
\end{align}
where $p(x_1|\theta,x_0)\equiv p(x_1)$.

Based on the role of the parameter as an auxiliary variable of the dynamic model, we introduce an alternative formulation assuming an extended distribution where $\bm z=\{\bm x_{1:T},\bm \theta_{1:T}\}$:
\begin{align}
\bar p(\bm x_{1:T},\bm \theta_{1:T}|\bm y_{1:T})
=&p(\bm x_{1:T},\theta_T|\bm y_{1:T})h_T(\bm \theta_{1:T-1}|\theta_T,\bm x_{1:T})\notag\\
\propto&p(\bm x_{1:T})p(\bm y_{1:T}|\bm x_{1:T})p(\theta_T|\bm x_{1:T})h_T(\bm \theta_{1:T-1}|\theta_T,\bm x_{1:T})
\label{eq:target_extended}
\end{align}
where $h_T$ is an arbitrary distribution.
Setting $\theta_T=\theta$ implies that \eqref{eq:target_extended} marginalizes to \eqref{eq:target_smc_nonMarkovian}; thus, by targeting $\bar p(\bm x_{1:T},\bm \theta_{1:T}|\bm y_{1:T})$, we can draw samples from $p(\bm x_{1:T},\theta|\bm y_{1:T})$. 
A particular interesting case is when $h_T(\bm\theta_{1:T-1}|\theta_T,\bm x_{1:T})=\prod_{s=1}^{T-1}p(\theta_s|\bm x_{1:s})$. In that case, 
\begin{align}
\bar p(\bm x_{1:T},\bm \theta_{1:T}|\bm y_{1:T})
\propto&\prod_{s=1}^Tp(x_{s}|\bm x_{<s})p(\theta_s|\bm x_{1:s})p(y_s|x_s)\notag\\
=&\prod_{s=1}^Tp(\theta_s|\bm x_{<s})p(x_{s}|\theta_s,x_{s-1})p(y_s|x_s)\label{eq:p.bar.rewritten},
\end{align}
where $\bm x_{<s-1} := (x_1, \dots, x_{s-1}) \text{ for } s \ge 2, \quad \bm x_{<1} := \varnothing$;
showing that the model can be written as a dynamic process with time-varying parameters. However, other choices of $h_T$ might also be of interest.

For construction of SMC algorithms, consider the sequence of unnormalized densities:
\begin{equation}
\gamma_t(\bm x_{1:t},\bm\theta_{1:t})
=
p(\theta_t)\left[\prod_{s=1}^tp(x_s|x_{s-1};\theta_t)p(y_s|x_s)\right]h_t(\bm\theta_{1:t-1}|\theta  _t,\bm x_{1:t})
\label{eq:pgas_prop}
\end{equation}
with $p(x_1|x_0,\theta_t)=p(x_1)$ and $h_1(\cdot)\equiv1$,
where $\{h_t\}$ is now a sequence of distributions with  $h_T(\cdot)$ from~\eqref{eq:target_extended}.  Similar to~\eqref{eq:p.bar.rewritten}, we may rewrite~\eqref{eq:pgas_prop} to
\begin{align}
\gamma_t(\bm x_{1:t},\bm \theta_{1:t})
\propto&\left[\prod_{s=1}^tp(\theta_s|\bm x_{<s})p(x_{s}|\theta_s,x_{s-1})p(y_s|x_s)\right]\cdot \frac{h_t(\bm\theta_{1:t-1}|\theta  _t,\bm x_{1:t})}{\prod_{s=1}^{t-1}p(\theta_s|\bm x_{1:s})},\label{eq:gamma.t.alt}
\end{align}
making it possible to utilize a formulation of time-varying parameters for any choice of auxiliary distributions $\{h_t(\cdot)\}$.

\subsection{Conditional SMC and ancestor sampling in the extended space}

Consider now a (conditional) SMC procedure where a proposal distribution $q_t(\theta_t,x_t|\cdot)$ is applied on both the parameter $\theta_t$ and the latent variable $x_t$. Defining $z_t=(x_t,\theta_t)$, 
the importance weights involved will then be
\begin{align}
w_t
=&
\frac{\gamma_t(\bm x_{1:t},\bm \theta_{1:t})}
     {\gamma_{t-1}(\bm x_{1:t-1},\bm \theta_{1:t-1})q_t(\theta_t,x_t|\cdot)}\notag \\
=&\frac{p(\theta_t)\left[\prod_{s=1}^tp(x_s|x_{s-1};\theta_t)\right]p(y_t|x_t)h_t(\bm\theta_{1:t-1}|\theta_t,\bm x_{1:t})
}{p(\theta_{t-1})\left[\prod_{s=1}^{t-1}p(x_s|x_{s-1};\theta_{t-1})\right]h_{t-1}(\bm\theta_{1:t-2}|\theta_{t-1},\bm x_{1:t-1})
q_t(\theta_t,x_t|\cdot)}.
\label{eq:feyn_general}
\end{align}
If we choose $h_t$ so that it factorizes in time for $t>1$:
\begin{align}
h_t(\bm\theta_{1:t-1}|\theta_t,\bm x_{1:t})=h_{t-1}(\bm\theta_{1:t-2}|\theta_{t-1},\bm x_{1:t-1}) \hbar_t( \theta_{t-1}|\theta_{t},\bm x_{1:t}),
\label{eq:h_factorization}
\end{align}
then \eqref{eq:feyn_general} simplifies further to:
\begin{equation}\label{eq:forward_weights}
\begin{split}
  w_t
=&\frac{p(\theta_t)\left[\prod_{s=1}^tp(x_s|x_{s-1};\theta_t)\right]p(y_t|x_t)
\hbar_t(\theta_{t-1}|\theta_t,\bm x_{1:t})}{p(\theta_{t-1})\left[\prod_{s=1}^{t-1}p(x_s|x_{s-1};\theta_{t-1})\right]
q_t(\theta_t,x_t|\cdot)}
\\
=&\frac{p(\theta_t|\bm x_{1:t-1})p(x_t|\theta_t,x_{t-1})p(y_t|x_t)\hbar_t(\theta_{t-1}|\theta_t,\bm x_{1:t})}
     {p(\theta_{t-1}|\bm x_{1:t-1})q_t(\theta_t,x_t|\cdot)}
\end{split}
\end{equation}
where the second expression is based on the alternative formulation~\eqref{eq:gamma.t.alt}.

While $h_t$ (or $\hbar_t$) may be chosen freely provided the dependency structure in \eqref{eq:target_extended} is respected, it can be viewed as a mechanism to tune the overall correlation structure along the time dimension.
For the specific choice of $\hbar_t=p(\theta_{t-1}|\bm x_{1:t-1})$, we get:
\begin{align*}
  w_t
=&\frac{p(\theta_t|\bm x_{1:t-1})p(x_t|x_{t-1},\theta_t)p(y_t|x_t)}{q_t(\theta_t,x_t|\cdot)}
\end{align*}
which simplifies further to $w_t=p(y_t|x_t)$ with the choice $q_t(\theta_t,x_t|\cdot)=p(\theta_t|\bm x_{1:t-1})p(x_t|x_{t-1},\theta_t)$.



In the extended space defined in \eqref{eq:target_extended}, assuming the unnormalized target \eqref{eq:pgas_prop}, the ancestor sampling (AS) weights correspond to:
\begin{align}
\label{eq:AS_weight_general}
\widetilde{w}_{t-1,T}^{(i)}&=w_{t-1}^{(i)}\frac{\gamma_{T}(\langle\bm x_{1:t-1}^{(i)},\bm x_{t:T}^*\rangle,\langle\bm \theta_{1:t-1}^{(i)},\bm \theta_{t:T}^*\rangle)}{\gamma_{t-1}(\bm x_{1:t-1}^{(i)},\bm \theta_{1:t-1}^{(i)})}\\
&=w_{t-1}^{(i)}\frac{p(\theta^*_{T})p(\langle x_{1:t-1}^{(i)},x^*_{t:T}\rangle|\theta_T^*)h_T(\langle\bm\theta^{(i)}_{1:t-1},\theta^*_{t:T-1}\rangle|\theta_T^* ,\langle\bm x^{(i)}_{1:t-1},x^*_{t:T}\rangle)\prod_{s=t}^{T} \left[ p(y_s|x_s^*)\right]}{p(\theta^{(i)}_{t-1})\left[\prod_{s=1}^{t-1}p(x^{(i)}_s|x^{(i)}_{s-1};\bm\theta^{(i)}_{t-1})\right]h_{t-1}(\bm\theta^{(i)}_{1:t-2}|\theta^{(i)}_{t-1},\bm x^{(i)}_{1:t-1})}.\notag
\end{align}
Assuming factorization of $h_t$ as in \eqref{eq:h_factorization}, we obtain:
\begin{equation}\label{eq:AS_weight}
\begin{split}
    \widetilde{w}_{t-1,T}^{(i)}\propto& w_{t-1}^{(i)} \frac{p(\theta_T^*|\bm x_{1:t-1}^{(i)})}{p(\theta_{t-1}^{(i)}|\bm x_{1:t-1}^{(i)})}p(x_t^*|\theta_T^*,\bm x_{t-1}^{(i)})\hbar_t(\theta_{t-1}^{(i)}|\theta_t^*,\langle\bm x_{1:t-1}^{(i)},x_t^*\rangle)\times\\
    &\prod_{s=t+1}^T\hbar_s(\theta_{s-1}^*|\theta_s^*,\bm x_{1:t-1}^{(i)},\bm x_{t:s}^*)
\end{split}
\end{equation}
Setting $\hbar_t=p(\theta_{t-1}|\bm x_{1:t-1})$ as before, gives:
\begin{align}
    \widetilde{w}_{t-1,T}^{(i)}&\propto w_{t-1}^{(i)} p(\theta_T^*|\bm x_{1:t-1}^{(i)})p(x_t^*|\theta_T^*, x_{t-1}^{(i)})\prod_{s=t}^{T-1}p(\theta_{s}^*|\langle\bm x_{1:t-1}^{(i)},\bm x_{t:s}^*\rangle)
    \label{eq:asweight_costly}.
\end{align}
As in the general non-Markovian case, calculating the product term $\prod_{s=t}^Tp(\theta_{s}^*|\langle\bm x_{1:t-1}^{(i)},\bm x_{t:s}^*\rangle)$ can be costly. 

In models where suitable sufficient statistics are available, the computational cost of estimating the full weights can be dramatically reduced.
It is also possible to exploit the flexibility of the setting by utilizing $h_t$ to impose Markov structure across the extended parameter space $\bm \theta_{1:T}$ leading to much simpler AS weights. 
Let for example $\hbar_t(\theta_{t-1}|\theta_t,x_{1:t})= f(\theta_{t-1}|\theta_t)$ for some $f$ so that $f(\theta_t|\theta_{t-1})=f(\theta_{t-1}|\theta_{t})$ (e.g. $f(\theta_{t-1}\mid\theta_t)
= \mathcal N\!\left(\theta_{t-1};\, \theta_t, \tilde\sigma^2\right)$ given some arbitrary variance $\tilde\sigma^2$).
For this specific choice the forward resampling weights \eqref{eq:forward_weights} become:
\begin{align*}
    w_t&=\frac{p(\theta_t) \left[\prod_{s=1}^t p(x_s|x_{s-1},\theta_t) \right]f(\theta_{t-1}|\theta_t)p(y_t|x_t)}{p(\theta_{t-1}) \left[\prod_{s=1}^{t-1} p(x_s|x_{s-1},\theta_{t-1}) \right]q(\cdot)}\\
    &=\frac{p(\theta_t|\bm x_{1:t-1})p(x_t|x_{t-1},\theta_t) f(\theta_{t-1}|\theta_t)p(y_t|x_t)}{p(\theta_{t-1}|\bm x_{1:t-1})q(\cdot)},
\end{align*}
which simplifies noticeably to 
\begin{align}
\label{eq:wfor_rw}
    w_t=\frac{p(\theta_t|\bm x_{1:t-1})}{p(\theta_{t-1}|\bm x_{1:t-1})}p(y_t|x_t)
\end{align}
when the proposal corresponds to $f(\theta_t|\theta_{t-1})p(x_t|x_{t-1},\theta_t)$.

More interestingly, the AS weights \eqref{eq:AS_weight} for this choice of $\hbar_t$ become:
\begin{align}
\label{eq:was_rw}
    \widetilde w_{t-1}^{(i)}\propto w_{t-1}^{(i)} \frac{p(\theta_T^*|\bm x_{1:t-1}^{(i)})p(x_t^*|\theta_T^*,x_{t-1}^{(i)})}{p(\theta_{t-1}^{(i)}|\bm x_{1:t-1}^{(i)})}f(\theta_{t-1}^{(i)}|\theta^*_t)
\end{align}
for which the computational complexity is then reduced to $\mathcal{O}(N)$.

\begin{algorithm}[t]
\DontPrintSemicolon
\caption{p-CSMC-AS}
\label{alg:p-CSMC}
\KwIn{Reference trajectory $\{\bm{x}^*_{1:T},\bm\theta_{1:T}^*\}$}
\KwOut{A new sample $\bm{x}_{1:T},\bm\theta_{1:T}$ distributed as  $\bar p(\bm{x}_{1:T},\bm\theta_{1:T}|\bm y_{1:T})=p(\bm x_{1:T},\theta_T|\bm y_{1:T})h_T(\bm\theta_{1:T-1}|\theta_T,\bm x_{1:T})$ }
Draw $x_{1}^{(i)},\theta_1^{(i)}\sim q_1(\cdot)$ for $i=1,\dots,N-1$ \\
Set $x_{1}^{(N)},\theta_1^{(N)}=x^*_{1},\theta_1^*$ \\
Set $w_{1}^{(i)}=\frac{p(x^{(i)}_1)p(\theta_1^{(i)}|x_1^{(i)})p(y_{1}|x^{(i)}_{1})}{q_1(x^{(i)}_{1},\theta_1^{(i)})}$ for $i=1,\ldots,N$ \\
\For{$t=2$ \KwTo $T$}{
Draw $a_t^{i}$ with $\mathbb{P}(a^{i}_t = k) \propto w^{(k)}_{t-1}, \quad \text{for } i = 1, 2, \ldots, N-1$ \\
Draw $(x_t^{(i)},\theta_t^{(i)})\sim q_t(x_t^{(i)},\theta^{(i)}_t|\cdot) ,\quad \text{for } i = 1, 2, \ldots, N-1$ \\
Set $(x_t^{(N)},\theta_t^{(N)})=(x_t^*,\theta_t^*)$\\
Compute $\widetilde w_{t-1|T}^{(i)}=w_{t-1}^{(i)}\frac{\gamma_T(\bm x_{1:t-1}^{(i)},\bm x_{t:T}^*,\bm\theta_{1:t-1}^{(i)},\bm\theta_{t:T}^*)}{\gamma_{t-1}(\bm x_{1:t-1}^i,\bm\theta_{1:t-1}^{(i)})}$ for $i=1,...,N$ as in \eqref{eq:AS_weight}\\
Draw $a_t^{N}$ with $P(a_t^{N}=k)\propto \widetilde w^{(i)}_{t-1|T}$ \\
Set $\bm x_{1:t}^{(i)}=(x_{1:t-1}^{(a_t^{i})},x_t^{(i)})$ and $\bm\theta_{1:t}^{(i)}=(\bm\theta_{1:t-1}^{(a_t^{i})},\theta_t^{(i)})$ for $i=1,...,N$\\
Set $w_t^{(i)}=\frac{\gamma_t(\bm x_{1:t}^{(i)},\bm\theta_{1:t}^{(i)})}{\gamma_{t-1}(\bm x_{1:t-1}^{(i)},\bm\theta_{1:t-1}^{(i)})q_t(x_t^{(i)},\theta_t^{(i)})}$
 as in \eqref{eq:forward_weights}}
Draw $k$ with $\mathbb{P}(k=i)\propto w_T^{(i)}$ \\
\textbf{return} $(\bm{x}^{(k)}_{1:T},\bm\theta_{1:T}^{(k)})$
\end{algorithm}

Algorithm~\ref{alg:p-CSMC} describes the p-CSMC algorithm with ancestor sampling 

\begin{proposition}
Algorithm \ref{alg:p-CSMC} leaves the target distribution $\bar p(\bm{x}_{1:T},\bm\theta_{1:T}|\bm y_{1:T})$ invariant (i.e. it defines a Markov transition kernel that preserves $\bar p(\bm{x}_{1:T},\bm\theta_{1:T}|\bm y_{1:T})$ as its stationary distribution.
\end{proposition}
\begin{proof}
Following \citet{andrieu2010particle}, we denote by $\phi$ the target distribution in the augmented space including all the random variables generated by the algorithm, which in this case can be defined as:
\begin{align}
\phi(\bm x_{1:T},\bm \theta_{1:T},\bm a_{2:T},k)&=\phi(\bm x^{b_{1:T}}_{1:T},\bm \theta^{b_{1:T}}_{1:T},b_{1:T}) \quad \phi(\bm x_{1:T}^{-b_{1:T}},\bm \theta_{1:T}^{-b_{1:T}},\bm a_{2:T}^{-b_{1:T}}|\bm x^{b_{1:T}}_{1:T},\bm \theta^{b_{1:T}}_{1:T}) \label{eq:augm_space} \\
&=\frac{\bar p(\bm x^{b_{1:T}}_{1:T},\bm \theta^{b_{1:T}}_{1:T}|\bm y_{1:T})}{N^T} \quad \phi(\bm x_{1:T}^{-b_{1:T}},\bm \theta_{1:T}^{-b_{1:T}},\bm a_{2:T}^{-b_{1:T}}|\bm x^{b_{1:T}}_{1:T},\bm \theta^{b_{1:T}}_{1:T})  \notag,
\end{align}.

The first term corresponds to the marginal distribution of interest, and the second is the conditional distribution from which the algorithm samples from to generate all the new proposals. By embedding the auxiliary parameters into some extended state space $\bm z_{1:t}=\{\bm x_{1:T},\bm \theta_{1:T}\}$, it then follows directly from the proof described in \citet{lindsten2014particle} that Algorithm \ref{alg:p-CSMC} leaves the target distribution invariant.
\end{proof}

\subsection{The partially collapsed p-CSMC-AS algorithm}

In the bootstrap implementation of the p-CSMC-AS algorithm, where the auxiliary parameters are sampled directly from the model, the forward weights do not depend on $\bm \theta_{1:T}^*$. This suggests the possibility of implementing a partially collapsed version of Algorithm~\ref{alg:p-CSMC}, in which the auxiliary parameters are resampled at each iteration but are excluded from the marginal distribution of interest.

This basically means that the intermediate targets, when $t<T$, are $p(\bm x_{1:t}|\bm y_{1:t})$, whereas the final target at time $T$ is the full joint $p(\bm x_{1:T},\theta_T|\bm y_{1:T})$.

While this subtle distinction makes no difference in terms of the forward weights, given that our effective proposal is the marginal $p( x_{t}|\bm x_{1:t-1})$, it simplifies the ancestor sampling weights noticeably.

More formally, consider now the (unnormalized) target distributions:
\begin{align*}
\gamma_t(\bm x_{1:t})=&p(\bm x_{1:t})p(\bm y_{1:t}|\bm x_{1:t}),&&t<T\\
\gamma_T(\bm x_{1:T},\theta_T)=&p(\bm x_{1:T},\theta_T)p(\bm y_{1:T}|\bm x_{1:T}),\\
\intertext{with the specific proposals}
q_t(x_t|\bm x_{1:t-1})=&p(x_t|\bm x_{1:t-1}),&&t<T\\
q_T(x_T,\theta_T|\bm x_{1:T-1})=&p(x_T|\bm x_{1:T-1})p(\theta_T|\bm x_{1:T}).
\end{align*}
Then the forward weights are
\begin{align}
w_t^{(i)}=&\frac{\gamma_t(\bm x_{1:t}^{(i)})}{\gamma_{t-1}(\bm x_{1:t-1}^{(i)})q_t(x_t|\bm x_{1:t-1})} =p(y_t|x_t)\notag\\
\intertext{while the ancestor weights become}
 \widetilde{w}_{t-1}^{(i)}=&w_{t-1}^{(i)}\frac{\gamma_T(\bm x_{1:t-1}^{(i)},\bm x_{t:T}^*,\theta_T^*)}{\gamma_{t-1}(\bm x_{1:t-1}^{(i)})}\notag\\
 =&w_{t-1}^{(i)}\frac{p(\bm x_{1:t-1}^{(i)},\bm x_{t:T}^*,\theta_T^*)p(\bm y_{1:T}|\bm x_{1:t-1}^{(i)},\bm x_{t:T}^*)}{p(\bm x_{1:t-1}^{(i)})p(\bm y_{1:t-1}|\bm x_{1:t-1}^{(i)})}\notag\\
 \propto&w_{t-1}^{(i)}p(\theta_T^*|\bm x_{1:t-1}^{(i)})p(x_t^*|x_{t-1}^{(i)},\theta_T^*).\label{eq:collapsed_weight}
\end{align}

\begin{algorithm}[h]
\DontPrintSemicolon
\caption{col-p-CSMC-AS}
\label{alg:col-p-CSMC}
\KwIn{Reference trajectory $\{\bm{x}^*_{1:T},\theta_{T}^*\}$}
\KwOut{A new sample $\bm{x}_{1:T},\theta_{T}$ distributed as  $p(\bm{x}_{1:T},\theta_{T}|\bm y_{1:T}) \equiv p(\bm x_{1:T},\theta|\bm y_{1:T})$ }
Draw  $x_1^{(i)}\sim p(x_1)$  for $i=1,\dots,N-1$ \\
Set $x_{1}^N=x^*_{1}$ \\
Set $w_{1}^{(i)}=p(y_{1}|x^{(i)}_{1})$ for $i=1,\ldots,N$ \\
\For{$t=2$ \KwTo $T$}{
Draw $a_t^{i}$ with $\mathbb{P}(a^{i}_t = k) \propto w^{(k)}_{t-1}, \quad \text{for } i = 1, 2, \ldots, N-1$ \\
Draw  $x_t^{(i)}\sim p(x_t|\bm x_{1:t-1}^{(a_t^{i})})$  \quad \text{for } $i = 1, 2, \ldots, N-1$ \\
Set $x_t^{(N)}=x_t^*$\\
Compute $\widetilde w_{t-1|T}^{(i)}=w_{t-1}^{(i)}\frac{p(\bm x_{1:t-1}^{(i)},\bm x_{t:T}^*,\theta_{T}^*,\bm y_{1:T})}{p(\bm x_{1:t-1}^{(i)},\bm y_{1:t-1})}$ for $i=1,...,N$ as in \eqref{eq:collapsed_weight}\\
Draw $a_t^{N}$ with $P(a_t^{N}=k)\propto \widetilde w^{(i)}_{t-1|T}$ \\
Set $\bm x_{1:t}^{(i)}=(x_{1:t-1}^{(a_t^{i})},x_t^{(i)})$ for $i=1,...,N$\\
Set $w_t^{(i)}=p(y_t|x^{(i)}_t)$}
Draw $k$ with $\mathbb{P}(k=i)\propto w_T^{(i)}$ \\
Draw $\theta_T \sim p(\cdot |\bm x_{1:T}^{(k)})$ \\
\textbf{return} $(\bm{x}^{(k)}_{1:T},\theta_T)$
\end{algorithm}

The core structure of this method, which we refer to as the partially collapsed pCSMC-AS (col-pCSMC-AS) is summarized in Algorithm \ref{alg:col-p-CSMC}.

Concerning step 6: In models where $p(x_t|x_{1:t-1})$ is not directly available, simulation can in practice be performed through the two step procedure
\begin{enumerate}
    \item $\theta_t^{(i)}\sim p(\theta_t|\bm x_{<t}^{(a_t^{i})})$;
    \item $x_t^{(i)}\sim p(x_t|x_{t-1}^{(a_t^{i})},\theta_t^{(i)})$.
\end{enumerate}

Note that, although it is obviously possible to go full marginal by marginalizing out $\theta_T$ as well \citep{wigren2019parameter}, retaining it preserves a conditional Markov structure over the latent space across the MCMC dimension, yielding much simpler ancestor sampling weights. Further, note that even if we require the possibility to simulate from $p(x_t|\bm x_{1:t-1})$, we do not need to compute the corresponding densities, only those conditional on $\theta_T^*$. 

 Although this algorithm may closely resemble a pGibbs strategy where only a single parameter sample is drawn at the end of each CSMC iteration, in the partially collapsed pCSMC-AS algorithm this parameter is not used in the proposal mechanism, which substantially reduces the correlation between iterations. 

 The methods defined in Algorithms \ref{alg:p-CSMC} and \ref{alg:col-p-CSMC} provide highly flexible frameworks in which unknown static parameters can be embedded into the state space to enable joint inference on states and parameters.  As shown on section \ref{sec-verify}, this can be specially convenient in complex problems by potentially grouping variables together in multiple ways.

\section{Experiments and applications}\label{sec-verify}

\subsection{The model}
A widely used framework for modeling infectious disease incidences in discrete time is based on stochastic branching processes, also referred to as renewal equation models. Let $x_t$ denote the number of newly infected individuals at time $t$ which is latent (unobserved). The expected number of new infections at time $t$ is then given by the product of a time-varying reproduction number $R_t$ and a weighted sum of previous incidence: 
\begin{equation}
x_t \sim \text{Poisson}\!\left( R_t \sum_{m=1}^{p} \theta_m x_{t-m} \right),
\label{eq:branching}
\end{equation}
for some initial seeding over the interval $t=-p+1$ and $t=0$ so that $\sum_{t=-p+1}^0x_t>0$.

We assume a dynamic model on $R_t$ which we define as a random walk on the log scale:
\[ \log(R_t) \sim \mathcal{N}(\log(R_{t-1}),\sigma_R^2)
\] where $\log(R_1)\sim \mathcal{N}(0,\sigma_R^2)$ for some known variance $\sigma_R^2$.

We additionally implement a binomial observational model on top that relates the latent space to actual observations $\bm y_{1:T}$ that in this specific case correspond to hospitalization incidences:
\begin{equation}\label{eq:hosp}
y_t \sim \mathrm{Binomial}(x_{t},p_{hosp})
\end{equation}
for some probability $p_{hosp}$.

Model~\eqref{eq:branching} admits a natural interpretation as a Poisson branching process where each individual infected at time $t-m$ independently generates secondary infections at time $t$ according to a Poisson distribution with mean $R_t \theta_m$. Summing over all infectious individuals yields the aggregate incidence process.

The parameters $\theta_1,\ldots,\theta_p$ represent the infectiousness profile over time since infection and satisfy $\theta_m \ge 0$. They are typically derived from the generation interval or serial interval distribution and quantify the relative contribution of past cases at lag $m$ to new infections at time $t$. The upper limit $p$ denotes the maximum infectious period considered, beyond which contributions to transmission are assumed negligible. 
These parameters are typically normalized so that $\sum_{m=1}^{p} \theta_m = 1$, ensuring that they describe only the relative infectivity profile, while the overall scale of transmission is captured entirely by $R_t$.

$R_t$ denotes the effective reproduction number at time $t$, defined as the average number of secondary infections generated by a typical infectious individual under prevailing epidemiological conditions. Allowing $R_t$ to vary over time enables the model to capture changes in transmission due to for example behavioral changes, non-pharmaceutical interventions or seasonal effects. 

Estimation of $R_t$ from incidence data using renewal-type models has become standard practice in real-time epidemic monitoring \citep{WallingaTeunis2004,Cori2013,Thompson2019}. Its statistical simplicity has facilitated its widespread use on multiple types of surveillance datasets. Extensions of the basic model incorporate reporting delays, overdispersion via negative binomial, and change-point structures in $R_t$ to improve robustness in low-incidence or rapidly evolving epidemic settings \citep{Thompson2019,Parag2021}.

The model formulation~\eqref{eq:branching} does not directly allow for simple updates of the parameters given the latent process $\bm x_{1:T}$. We therefore consider a reformulation of the model where, for lags $m=1,\dots,p$, we introduce the per-lag counts or contributions $x_{t,m}$ such that the total incidence at time $t$ is $x_t=\sum_{m=1}^p x_{t,m}$ with
\begin{align*}
x_{t,m}\sim\text{Poisson}(R_t\theta_m\sum_{n=1}^px_{t-m,n}).
\end{align*}
We further assume independent Gamma priors on the parameters:
\[
\theta_m \sim \mathrm{Gamma}(\alpha_m,\beta_m).
\]
We then get that the conditional distribution $p(\theta_m|\bm R_{1:t},\bm x_{1:t,1:p})$ is available in closed form (see Section \ref{supl:subsec:gamma-pois} in the Supplementary Material (SM)) and corresponds to conditionally independent Gamma distributions:
\[
\theta_m | \bm x_{1:t,1:p}, \bm R_{1:t}, \bm y_{1:t}
\,\sim\, \mathrm{Gamma}\!\big(\alpha_m + C_{t,m},\; \beta_m + E_{t,m}\big),
\qquad m=1,\dots,p.
\]
Here the sufficient statistics per-lag can be iteratively updated as:
\[
C_{t,m} := \sum_{s=1}^t x_{s,m},
\qquad  
E_{t,m} := \sum_{s=1}^t \big(R_s\, x_{s-m}\big).
\]

Here, we present results from two different experiments within this framework: 
\begin{itemize} \item A simulation study based on synthetic data, where we assume $\bm R_{1:T}$ to be known and focus on inferring the infectivity profile $\bm\theta_{1:p}$. We compare the performance of PGAS, pCSMC-AS, col-pCSMC-AS, and a full marginal sampler in two scenarios of different complexity, with $p=2$ and $p=7$. 
\item An application in which we estimate both the infectivity profile and the reproductive number based on real data (hospitalizations) from the Alpha variant outbreak of SARS‑CoV‑2 in Norway. In this case, we compare only the performance of col-pCSMC-AS against a slice sampler (JAGS \citep{hornik2003jags}), since PGAS failed to converge within a reasonable time frame. \end{itemize}

\subsection{Simulation study}

In this first set of experiments with synthetic data, we assume $R_t$ to be known and focus on recovering the weights $\bm\theta_{1:p}$. Note that this makes the problem identifiable without additional constraints in terms of scale. 

Synthetic data was generated following model~\eqref{eq:branching}-\eqref{eq:hosp}.
We first predefined the number of dimensions $p$ and the time period $T$. Then, the parameters $\theta_m$ were sampled randomly from independent Gamma distributions $\text{Gamma}(\alpha_m,\beta_m)$, where $\alpha_m=2$ and $\beta_m=4$, and then normalized. $R_t$ was generated through a random walk on the log scale so that the resultant outbreak was sufficiently relevant (i.e. with some relevant incidence during the study period). And finally, we produced some synthetic observations through the observational model where we have used $p_{hosp}=0.1$. We set $x_{1}=10$ initially to seed the model. We assume $x_t=0$ for $t\leq0$

Two different settings, $p=2,T=50$ and $p=7,T=100$, were considered (see Figure \ref{supl:fig:syn_data} in the SM for the realised observation processes).
We compared 4 different algorithms: PGAS, pCSMC-AS, col-pCSMC-AS and the full marginal sampler. We have used bootstrap proposals in all the algorithms. 

In PGAS, we alternatingly sampled the latent space $p(\bm x_{1:T,1:p}|\theta_m,\bm R_{1:T},\bm y_{1:T})$ through a CSMC algorithm and the parameters externally through the conditionals $p(\theta_m|\bm x_{1:T,1:p},\bm R_{1:T})$. In the pCSMC-AS, we have set $\hbar_t=p(\theta_{t-1}|\bm x_{1:t-1})$ and $q_t(\theta_t,x_t|\cdot)=p(\theta_t|\bm x_{1:t-1})p(x_t|x_{t-1},\theta_t)$. 
In this specific model, the marginal transition distribution $p(x_{t,m}|\bm x_{1:t-1,m})$ is available in closed form (see section \ref{supl:sec:full_marginal} in the SM). This makes the marginal sampler~\citep{wigren2019parameter}, targeting the marginal distribution $p(\bm x_{1:T,1:p}|\bm R_{1:T},\bm y_{1:T})$ directly, possible to apply in this case. Due to the non-Markovian structure, the computational burden for this sampler is much higher when combined with ancestor sampling.

\subsubsection{Results}

For all the experiments we have assumed a Gamma prior for $\theta_m$ with hyperparameters $\alpha_m=2$ and $\beta_m=4$ (the same was used for generating the data). 
For all the algorithms, $N=300$ particles were used. 10.000 MCMC iterations were run, from which 500 represented the burn-in period (except for PGAS for which 40.000 MCMC iterations were run where the first 4.000 were discarded as burn-in). Further details on the test runs are given in section \ref{supl:sec:tests summaries} in the appendix.

\begin{figure}[H]
  \centering
    \includegraphics[width=\linewidth]{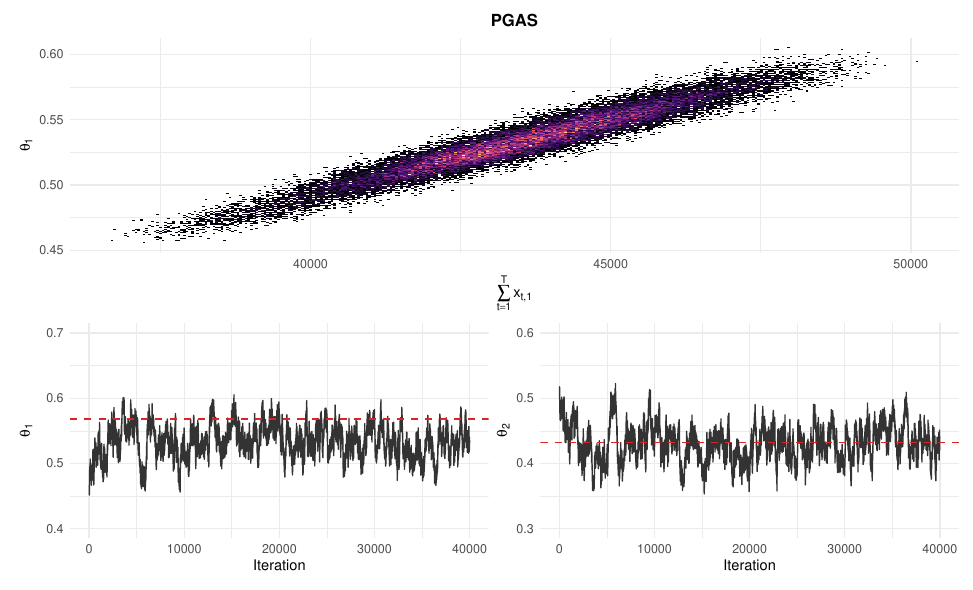}%
 
  \caption{Posterior samples and trace plots for $T=50$ and $p=2$ with PGAS. The upper panel shows the relation between the summary statistic $\sum_{t=1}^Tx_{t,1}$ and $\theta_1$ in the posterior samples obtained after each iteration. The lower panel shows traceplots of $\theta_1$ and $\theta_2$. The burn-in period for the posterior samples was 4000. Additional information on the setup for this example is available in the Supplementary material. }
  \label{fig:pgas-2d-panel}
\end{figure}

Figures~\ref{fig:acf_theta_d7} ($p=7$) and~\ref{supl:fig:acf_theta_d2} ($p=2$) show traceplots and autocorrelation functions for all the four algorithms.
We first consider the PGAS algorithm.
It quickly became evident that Gibbs sampling faced substantial difficulties within this framework due to the very strong internal correlations involved, in particular between the parameters and the latent space. In the simplest case with $p=2$, mixing was very slow in the parameter dimension, and the algorithm required a large number of samples to converge. In the more realistic experiment with $p=7$, achieving full convergence was infeasible in useful time, as shown in the upper left panel of Figure \ref{fig:acf_theta_d7}. 

On the other hand, both the pCSMC-AS and the col-pCSMC-AS algorithms, showed good mixing performance in the parameter space in both settings ($p=2$ and $p=7$), with the latter one slightly better, as illustrated by figures \ref{fig:acf_theta_d7} and \ref{supl:fig:acf_theta_d2}. As expected, the full marginal algorithm is the one that performs best in terms of mixing given that no correlation due to some fixed $\theta$ is carried on from one iteration to the next. Note however that the marginal algorithm has a higher computational cost.
\begin{figure}[H]
  \centering
    \includegraphics[width=0.9\linewidth]{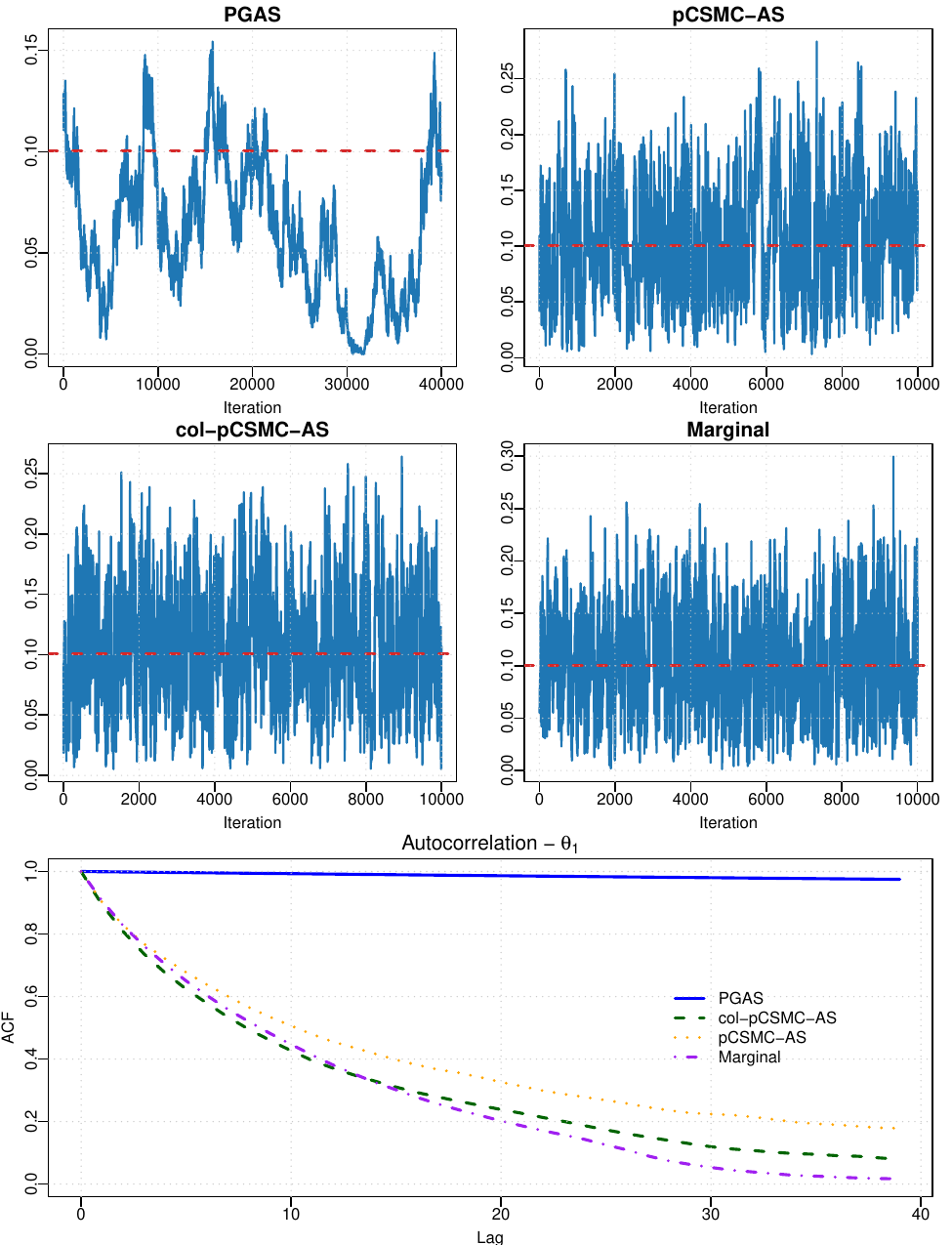}%
  \caption{Traceplots and autocorrelation plots of $\theta_1$ in the experiment with $p=7,T=100$. Burn-in periods are 4.000 for the PGAS and 500 for the rest (which are included in the trace plots but not on the acf plots). Real value is shown in the trace plots as red dashed lines for reference. Estimates for the marginal model were obtained by sampling from the conditional distributions at each iteration.}
  \label{fig:acf_theta_d7}
\end{figure}

We noticed some mixing problems on the first time points of the latent space, specially in the pCSMC-AS algorithm, as shown in Figures \ref{supl:fig:acf_latent_d2} and \ref{supl:fig:acf_latent_d7} . This is confirmed by the plot of the Effective Sample Size (ESS) of the marginal distribution of $x_t$ (Figures \ref{supl:fig:ess_d2} and \ref{supl:fig:ess_d7}) which was very low for small $t$, and a comparison of the Expected Jumping Distance (EJD) across methods. This problem, as previously described by \citet{wigren2019parameter}, also affects the full marginal sampler and is a consequence of the use of uninformative proposals at the first time points. 

The partially collapsed version of the pCSMC-AS algorithm emerges as an optimal compromise in this setting: it achieves performance comparable to that of the full marginal sampler while also incurring a much lower computational cost ($\mathcal{O}(NT)$ evaluations vs $\mathcal{O}(NT^2)$) thanks to the simpler ancestor sampling weights, with lower autocorrelation at the early time points of the latent space compared to p-CSMC-AS.

\subsection{Combined Estimation of the Reproductive Number and Infectivity Profile of the Alpha SARS‑CoV‑2 Variant in Norway}\label{subsec:covid_data}

We extend the approach from the previous section to address the more complex problem of estimating both the reproductive number and the infectivity profile simultaneously from real SARS-CoV-2 data, framing the analysis in the growing phase of the Alpha variant in Norway.
Several other parameters involved are assumed known.
In addition to an analysis based on real data, we also include in section \ref{supl:sec:combined} of the SM an example of the application of the same inference strategy to synthetic data for which the ground truth of both $\bm R_{1:T}$ and $\theta_{1:p}$ was known, as a proof of concept of the method. 

We approached the inference problem by iteratively targeting $p(\theta_{1:p},\bm x_{1:T,1:p}|\bm R_{1:T},\bm y_{1:T})$ with a col-pCSMC-AS algorithm, and $p(\bm R_{1:T}|\theta_{1:p},\bm x_{1:T,1:p},\bm y_{1:T})$ through a standard CSMC-AS routine. For this more complex setting, we were not able to make the other SMC algorithms to fully converge within a reasonable timeframe. We therefore only made a comparison with a JAGS implementation of the model.

Because we assume independent priors for each $\theta_m$, we fixed the scale by setting $\theta_p=0.05$ to make the combined inference of $\{\bm R_{1:T},\theta_{1:p},\bm x_{1:T}\}$ identifiable.

\subsubsection{The data}

The observations $\bm y_{1:T}$ are daily hospital covid-19 admissions in Norway during the period between February $1^{st}$ 2021 and March $15^{th}$ 2021, which corresponds to the growing phase of the Alpha (B.1.1.7) variant outbreak (see Figure \ref{fig:post_fit_covid}). 
The data was originally gathered by the \href{https://www.fhi.no/ss/korona/koronavirus/norsk-beredskapsregister-for-covid-19/}{\textit{Emergency preparedness register for COVID-19}} (Beredt C19) and subsequently made available for research purposes.

\subsubsection{Results}

Given that the mean Serial interval for the Alpha variant is expected to lay somewhere between 2 to 5 days (3.47 days, 95\% CI: 2.52–4.41, and its generation time between 4 and 5 days (4.35, 95\% CI: 3.91–4.8 \citep{xu2023assessing}), we set $p=8$ to fully capture the most relevant period. We assumed an overall hospitalization probability given infection of 0.047 (4.7\%), in line with previous studies \citep{nyberg2021risk}. With respect to the variance of the random walk model for $R_t$ on the log scale, we have used $\sigma_R=0.15$. 
We initialized the model by assuming $x_1\sim \text{Pois}(100)$ and $x_t=0$ for $t<1$. This is consistent with the expected number of infections implied by the observed hospital incidence at the start of the considered time-period. 
We ran a total of 90.000 iterations of the pCSMC-AS algorithm including 5.000 burnin samples. For JAGS, we have drawn 290.000 samples including 80.000 as burnin.

\begin{figure}[H]
  \centering
    \includegraphics[width=\linewidth]{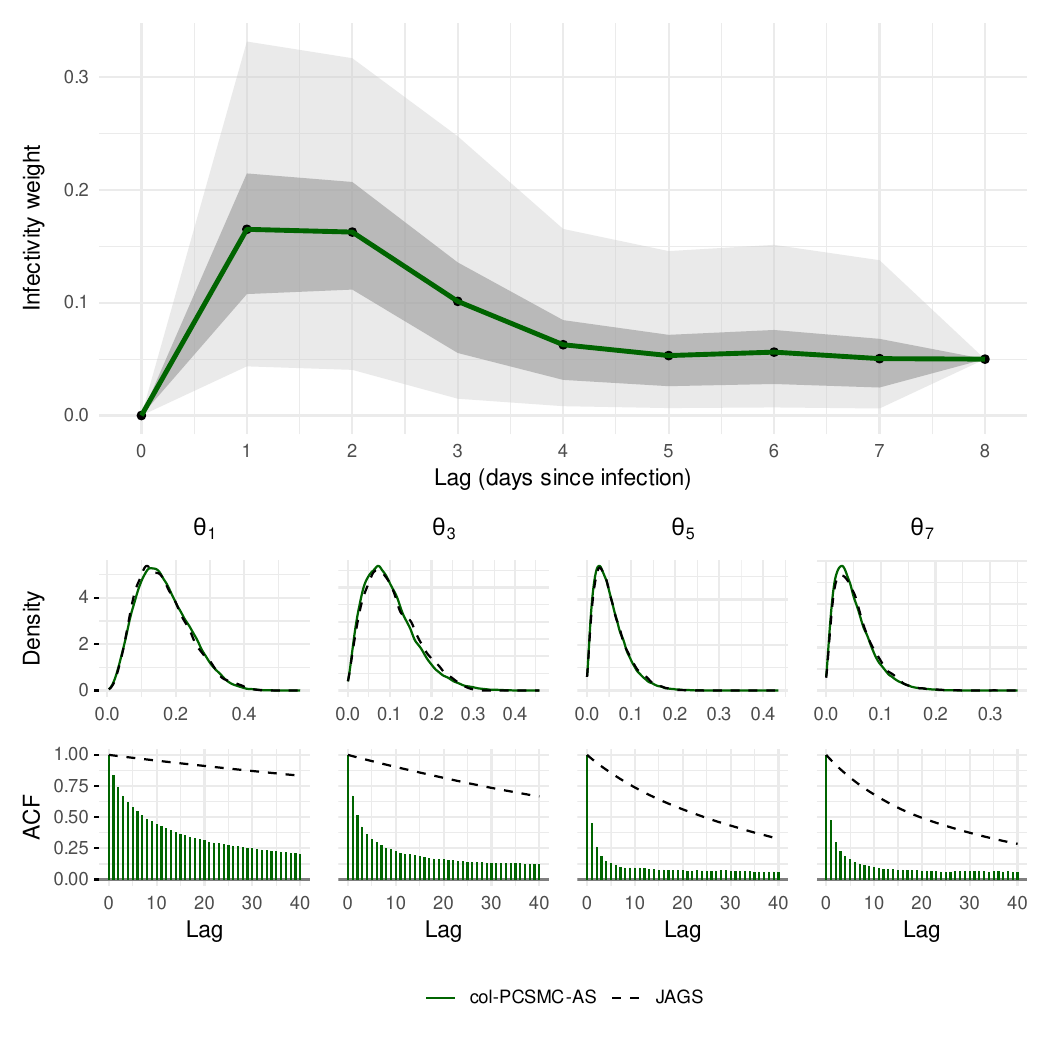}%
  \caption{Posterior distribution of $\theta_{1:p}$, and autocorrelation function comparing the performance of JAGS (slice sampler) and the partially collapsed PCSMC-AS algorithm. In the upper panel, all the posterior distributions are shown together so that $\theta_1$ corresponds to Lag 1, $\theta_2$ corresponds to Lag 2 and so on.}
  \label{fig:covid:post_theta}
\end{figure}

The posterior distribution of $\bm\theta_{1:p}$, as shown in Figure \ref{fig:covid:post_theta}, yielded an estimated average hospitalization interval of 3.4 days (2.8-4.1), which aligns quite well with the expected distribution of the generation time given the uncertainties at hand and the fact that the observed interval is expected to be shorter in growing phases of the epidemic \citep{park2021forward}.

\begin{figure}[H]
  \centering
    \includegraphics[width=\linewidth]{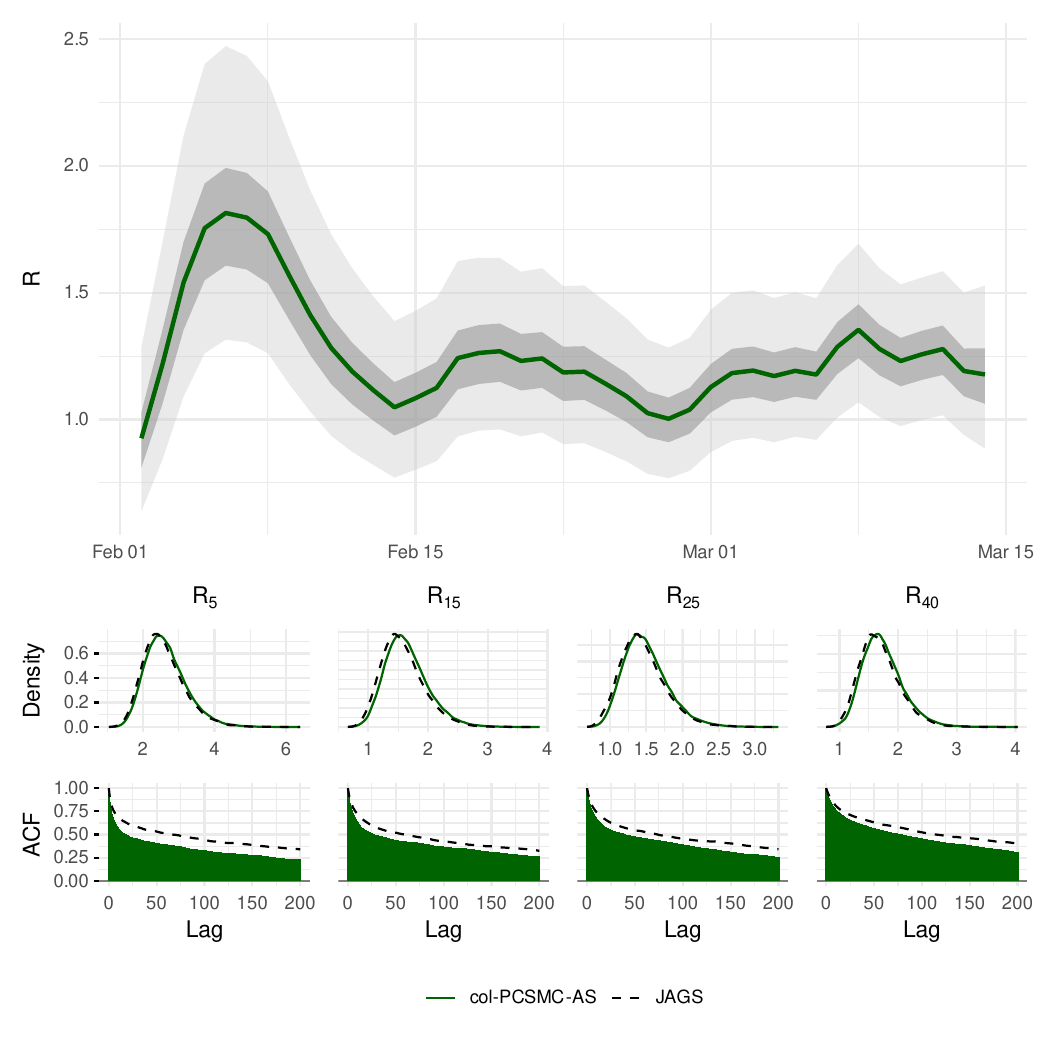}%
  \caption{Top figure shows the posterior distribution of $R_t$ with both the interquartile range and the 95\% credible interval, normalized as $R_t \sum_{m=1}^p\theta_m$. Density plots and autocorrelation functions are shown below comparing the mixing performance of the partially collapsed PCSMC-AS algorithm against JAGS (slice sampler).}
  \label{fig:covid:post_R}
\end{figure}

Even though the posterior distribution of the infectivity profile may not be very informative because of all the uncertainties involved, the method provides the posterior distribution of the reproductive number without making any specific assumptions on the distribution of the hospitalisation/serial interval. Figure \ref{fig:covid:post_R} shows the posterior distribution of $R_t$ obtained through the col-pCSMC-AS algorithm. 

We see in the autocorrelation plots that, in general, and specially in the case of the parameters, the mixing performance is much better with the partially collapsed pCSMC-AS than with the slice sampler (JAGS). In this regard, it is worth noting that it is only the latent space and the weights $\theta$ that are inferred through the col-pCSMC-AS algorithm, whereas $R_t$, is updated through a separate CSMC step. This largely explains the different mixing in both dimensions, compared to the slice sampler. The results of this comparison against slice sampling should not be generalized, as they are specific to this particular case and our goal is to compare performance with alternatives within the Particle Gibbs framework. The JAGS results are shown mainly as a confirmation of the posterior distribution obtained.

The goodness of fit and the uncertainty coverage of the model against real observations, as shown in Figure \ref{fig:post_fit_covid}, is very good.

\begin{figure}[H]
 \centering
    \includegraphics[width=0.75\linewidth]{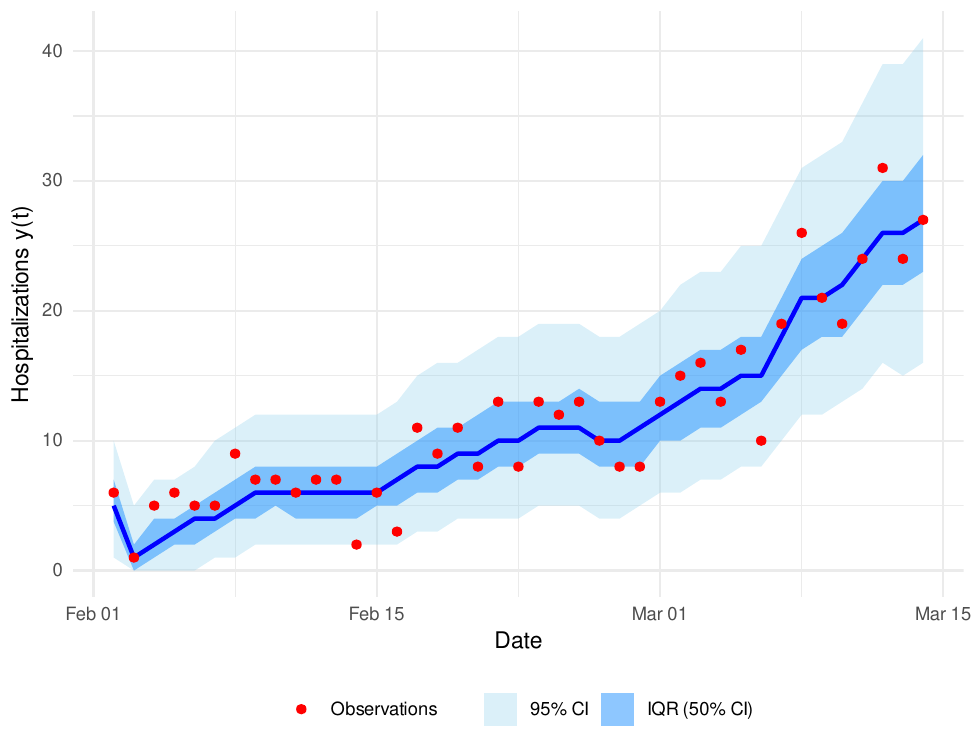}%
  \caption{Posterior fit including median and credible intervals of simulated values against real observations, obtained from 10.000 posterior samples of the col-pCSMC-AS algorithm.}
  \label{fig:post_fit_covid}
\end{figure}

\section{Conclusion}\label{sec-conc}

In the presence of unknown static parameters, PGAS is the preferred option whenever the conditional distribution of the parameters given the latent variables is available. However, in practice, its applicability is constrained to settings in which no strong internal correlations are involved since this reduces mixing to levels that render the method useless. Other common alternatives such as PMMH require a thorough choice of proposals which can be really challenging in high dimensional spaces. And even gradient based alternatives are not well defined in models with discrete latent variables.

In this work we have defined a generalized framework of CSMC algorithms (pCSMC-AS) for the combined inference of static unknown parameters and latent variables in state space models by extending the existing particle learning methodology beyond the online setting. We have showed how, for a certain type of models, it can outperform PGAS and, given its high flexibility, become another useful tool to consider when addressing inference in environments where mixing can be challenging.

In order to avoid particle degeneracy in CSMC-based algorithms, incorporating some form of backward sampling strategy is indispensable. However, this comes at a high computational cost, especially in non-Markovian settings. With our partially collapsed pCSMC-AS algorithm, we have shown that it is possible to mitigate this problem without compromising the efficiency of proposal generation in the forward sweep, leading to a much better mixing. 

Recent work showed that the upper bound of the mixing time, defined as the number of updates needed to produce a sample that is $\epsilon$ close in distribution to the target, for the CSMC-AS algorithms is $\mathcal{O}(T)$ \citep{lee2020coupled} or $\mathcal{O}(\log T)$ \citep{karjalainen2025mixing} under different strong mixing assumptions. However, these assumptions do not hold in many real life applications. It is well known that the performance of these backward sampling strategies deteriorates with models involving weakly informative observations, slowly mixing dynamics \citep{karppinen2024conditional} and, more importantly, when there is non-Markovian structure in the latent space. In these cases, AS/BS may alleviate the degeneracy problem to some extent but the particle system would still collapse \citep{lindsten2014particle} leading to poor mixing at the first time points. 

By contrast, when a substantial amount of information is transferred across iterations (e.g. in PGAS), the proposals for the first time points are generally better guided, which reduces the variance of the weights and, consequently, the severity of degeneracy. This improvement comes at a cost, however, in terms of poorer MCMC mixing in the presence of strong internal correlations, and a greater risk of getting trapped in local optima. Thus, there is an intrinsic trade-off between reducing particle degeneracy within the CSMC sampler by generating informed proposals and maintaining good mixing in the outer MCMC algorithm. 

We have defined a very general framework in which  intermediate auxiliary parameters can be used to generate proposals through different strategies, determined by the choice of distributions in the extended space, in a similar fashion to how intermediate target functions are used in twisted models. 

In addition, the pCSMC-AS algorithms indirectly exploits the ancestor sampling step as an intermediate approach between the inefficient scenario of generating trajectories that are independent from the reference and the Particle Gibbs scheme, in which all trajectories are coupled through a single common parameter which, as mentioned before, can lead to poor mixing in certain situations.

Possible future lines of work include, among others: exploring potentially optimal choices of auxiliary variables and distributions for different types of applications; considering alternative resampling schemes beyond multinomial resampling which may improve the performance specially in the presence of low informative observations \citep{karppinen2024conditional}; and investigating the incorporation of local proposals, as in \citet{malory2021bayesian,finke2023conditional,corenflos2024particle}, just to name a few alternatives.

\section{Code}\label{sec-code}
Code for all numerical simulations is available at
\href{https://github.com/Adizlois/pCSMC-AS}{\texttt{https://github.com/Adizlois/pCSMC-AS}}.

\section{Disclosure statement}\label{sec-disclosure-statement}
The authors declare that they have no conflicts of interest.

\section{Declaration of Generative AI Use}
The authors used generative AI (ChatGPT versions 5 to 5.6) solely to check spelling and improve the clarity of the text. No generative AI was used to generate scientific content, analyses, or results.

\section{Acknowledgements}
This work was supported by the Research Council of Norway, Integreat - Norwegian Centre for knowledge-driven machine learning, project number 332645.

\section{Supplementary Material}
Title: \textit{
Supplementary material for 'Parameter estimation in Conditional Sequential Monte Carlo algorithms through Particle Learning'}.

Description: It includes some derivations of the conjugate models, a description of the different settings of the experiments and detailed plots to inform about the performance of the algorithms. In addition, we include results from an additional experiment with synthetic data where both $R_{1:T}$ and $\theta$ are estimated.

\phantomsection\label{supplementary-material}
\bigskip




\newpage

\bibliographystyle{chicago}
\bibliography{bibliography}

\setcounter{section}{0}
\renewcommand{\thesection}{S.\arabic{section}}

\section*{Supplementary material for 'Parameter estimation in Conditional Sequential Monte Carlo algorithms through Particle Learning'}

\renewcommand{\thefigure}{S.\arabic{figure}}
\setcounter{figure}{0}

\section{Gamma-Poisson conjugacy}\label{supl:subsec:gamma-pois}

Because of the conjugacy between the Gamma and Poisson distributions we get that, the conditional distribution $p(\theta_m|R_t,x_{1:T,1:p})$ is available in close form:
\begin{align}
p(\theta_m | \bm x_{1:T,1:p}, \bm R_{1:T}, \bm y_{1:T}) 
&\propto p(\theta_m)\;
   p\big(\bm x_{1:T,1:p} | \theta_{1:p}, \bm R_{1:T}\big)\;
   p\big(\bm y_{1:T}|\bm x_{1:T}\big) \notag \\
&\propto p(\theta_m)\;
   \prod_{t=1}^T p\big(x_{t,m}| \theta_m, R_t, x_{t-m}\big) \notag \\
&\propto \left[\theta_m^{\alpha_m-1} e^{-\beta_m \theta_m}\right]
   \times \prod_{t=1}^T
     \left[\frac{(\theta_m \lambda_{t,m})^{x_{t,m}}}{x_{t,m}!}\,
           e^{-\theta_m \lambda_{t,m}}\right] \notag \\
&\propto \theta_m^{\alpha_m-1 + \sum_{t=1}^T x_{t,m}}\,
    \exp\!\Big(-\theta_m \big[\beta_m + \sum_{t=1}^T \lambda_{t,m}\big]\Big) \label{eq:Gam-Pois}
\end{align}

Where $\lambda_{t,m}:=R_t\, \sum_{i=1}^p x_{t-m,i}$
Define the sufficient statistics per-lag
\[
C_{T,m} := \sum_{t=1}^T x_{t,m},
\qquad  
E_{T,m} := \sum_{t=1}^T \lambda_{t,m}
     = \sum_{t=1}^T \big(R_t\, \sum_{i=1}^p x_{t-m,i}\big).
\]
Recognizing the Gamma kernel, the full conditional is then:
\[
\theta_m | \bm x_{1:T,1:p}, \bm R_{1:T}, \bm y_{1:T}
\,\sim\, \mathrm{Gamma}\!\big(\alpha_m + C_{T,m},\; \beta_m + E_{T,m}\big),
\qquad m=1,\dots,p.
\]

\section{Full marginal sampler}\label{supl:sec:full_marginal}

Let's define $x_{t}$ as the total number of counts at some time point $t$:
\[
x_{t} = \sum_{j=1}^p x_{t,j}
\]
We know that each per-lag count is Poisson distributed
\[
x_{t,m} | \theta_m, \bm x_{1:t-1,1:p}
\sim \text{Poisson}\!\big( R_t \,\theta_m\, x_{t-m} \big)
\]
We want the marginal distribution
\[
p(x_{t,m} | \bm x_{1:t-1,1:p}) =
\int_0^\infty 
  p(x_{t,m}|\theta_m,\bm x_{1:t-1,1:p})\,
  p(\theta_m|\bm x_{1:t-1})\,d\theta_m
\]
Assuming a Gamma prior $(\alpha,\beta)$ on the parameter:
\[
p(\theta_m) = \frac{\beta^{\alpha}}{\Gamma(\alpha)}\,
              \theta_m^{\alpha - 1} e^{-\beta \theta_m}.
\]
We get that:
\[
p(x_{t,m} | \bm x_{1:t-1,1:p})
= \int_0^\infty 
    \frac{e^{-R_t\theta_m x_{t-m}}(R_t\theta_m x_{t-m})^{x_{t,m}}}{x_{t,m}!}
    \frac{\beta^{\alpha}}{\Gamma(\alpha)}\,
    \theta_m^{\alpha - 1} e^{-\beta \theta_m}
  \, d\theta_m
\]

Gather terms that depend on $\theta_m$ inside the integral:
\[
p(x_{t,m} | \bm x_{1:t-1,1:p})
= \frac{(R_t x_{t-m})^{x_{t,m}}\,\beta^{\alpha}}{x_{t,m}!\,\Gamma(\alpha)}
  \int_0^\infty 
    \theta_m^{x_{t,m} + \alpha - 1}
    e^{-(\beta + R_t x_{t-m})\theta_m}
  \, d\theta_m.
\]
We recognize the functional form of a Gamma distribution. Therefore:
\begin{align*}
p(x_{t,m} | \bm x_{1:t-1,1:p})
&= \frac{\Gamma(x_{t,m} + \alpha)}{\Gamma(\alpha)\,x_{t,m}!}
  \frac{\beta^{\alpha}\,(R_t x_{t-m})^{x_{t,m}}}
       {(\beta + R_t x_{t-m})^{x_{t,m} + \alpha}}   \\
       &=\binom{x_{t,m} + \alpha - 1}{x_{t,m}}\, \big(\frac{\beta}{(\beta + R_t x_{t-m})}\big)^\alpha \big(\frac{R_t x_{t-m}}{(\beta + R_t x_{t-m})}\big)^{x_{t,m}} 
\end{align*}

Which corresponds to the pmf of a Negative Binomial with parameters: $r=\alpha$ and $p = \frac{\beta}{\beta + R_t x_{t-m}}$

\section{Tests with synthetic data: Inference on $\theta$} \label{supl:sec:tests summaries}

\begin{table}[H]
    \centering
    \begin{tabular}{lccc|ccc}

        & \multicolumn{3}{c|}{$p = 2$} & \multicolumn{3}{c}{$p = 7$} \\
        Method & Iterations & Burn-in & $N$ & Iterations & Burn-in & $N$ \\
        \hline
        PGAS                & 40.000  & 4.000  & 300  & 40.000  &  - & 300  \\
        pCSMC               & 10.000  & 500  & 300  & 10.000  & 500  & 300  \\
        pCSMC (collapsed)   & 10.000  & 500  & 300  & 10.000  & 500  &  300 \\
        Marginal            &10.000   & 500  & 300  & 10.000  & 500  & 300  \\
        \hline
        JAGS & 50.000 & 50.000 & - & 200.000 & 50.000 & - \\
        \hline
    \end{tabular}
    \caption{Simulation settings for $p=2$ and $p=7$. Three independent chains were run for each configuration $N$ corresponds to the number of particles used. An initial seeding of $2p$ cases was used at time t=0 in all the models.}
    \label{supl:tab:settings}
\end{table}

\begin{figure}[H]
  \centering
  \begin{minipage}[b]{0.48\textwidth}
    \centering
    \includegraphics[width=\textwidth]{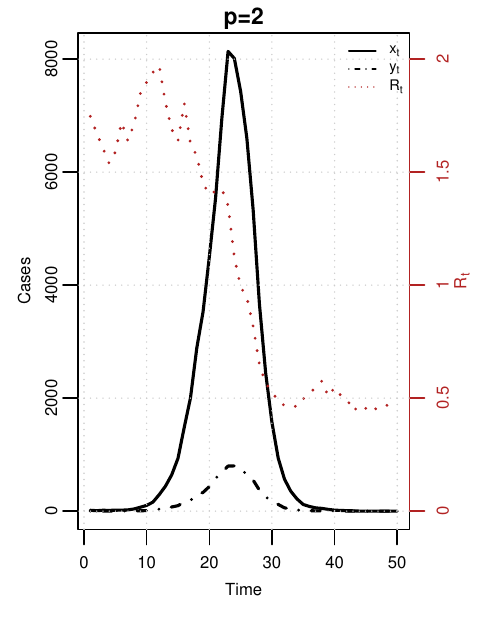}
    \end{minipage}
  \begin{minipage}[b]{0.48\textwidth}
    \centering
    \includegraphics[width=\textwidth]{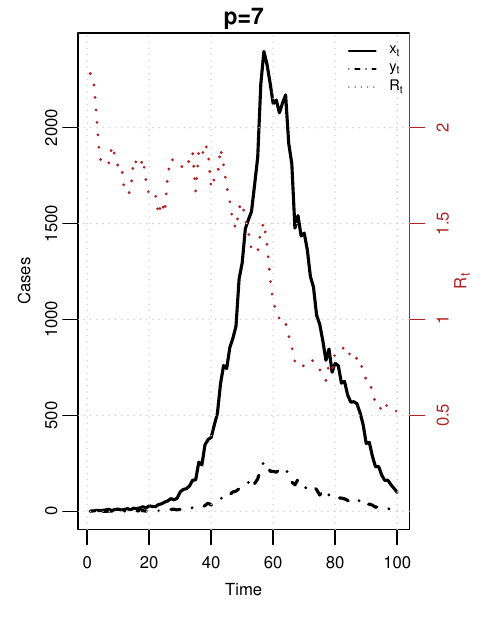}
    \end{minipage}
    \caption{Summary of the synthetic data generated for the experiments. Here $x_t$ corresponds to the number of cases, $y_t$ the observed counts and $R_t$ the (assumed known) reproduction number. The observational model uses $p_{hosp}=0.1$. The true generative parameters were (reported here as rounded values)
    $\theta_{real}=\{0.568,0.432\}$ and $\theta_{real}=\{0.100,0.119,0.183,0.114,0.116,0.175,0.192\}$ for the synthetic data with $p=2$ and p=7 respectively.}
    \label{supl:fig:syn_data}
\end{figure}
\newpage
\subsection{p=2}
\begin{figure}[H]
  \centering
\includegraphics[width=0.85\linewidth]{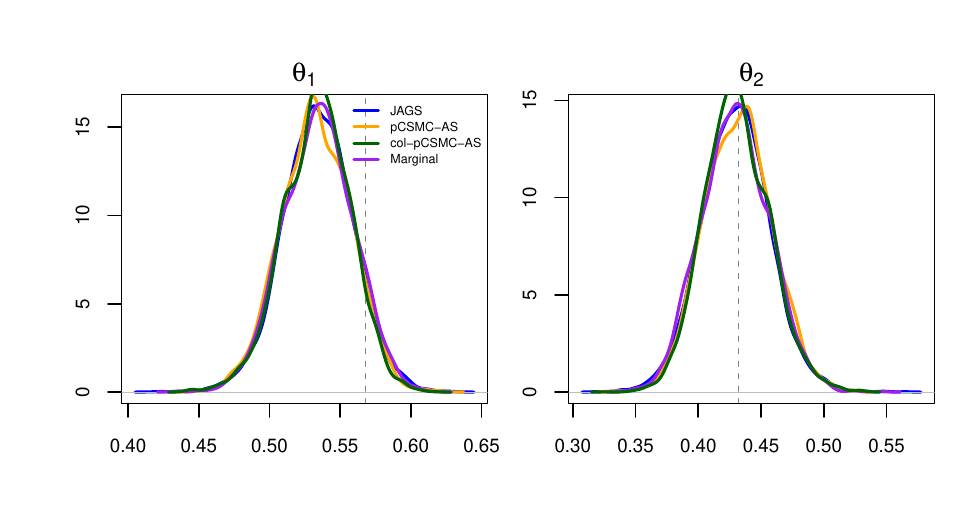}%
  \caption{Comparison of the posterior samples of the marginal distribution of the parameters for p=2 obtained with pCSMC-AS, col-pCSMC-AS and the full marginal CSMC against JAGS (slice sampler). The real values used in the generative model are shown as vertical dotted lines for reference. Parameter samples for the marginal model were obtained by sampling from the conditional distirbution available after each iteration.}
\label{supl:fig:posteriors_d2}
\end{figure}
\begin{figure}[H]
  \centering
\includegraphics[width=0.85\linewidth]{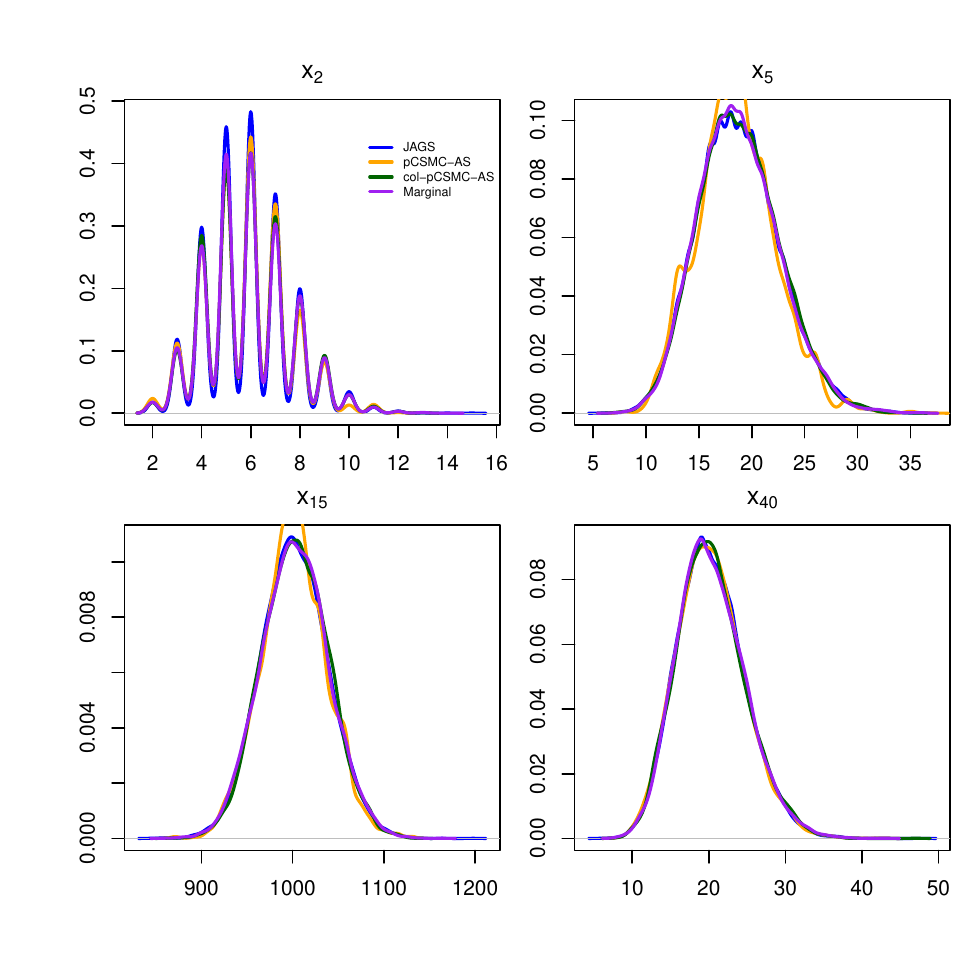}%
  \caption{Comparison of the marginal posterior distributions of $x_t=\sum_{m=1}^px_{t,m}$ for $t=\{2,5,15,40\}$ obtained with pCSMC-AS, col-pCSMC-AS and the full marginal CSMC against JAGS with $p=2$}
\label{supl:fig:posteriors_latentd2}
\end{figure}
\newpage
\begin{figure}[H]
  \centering
\includegraphics[width=0.9\linewidth]{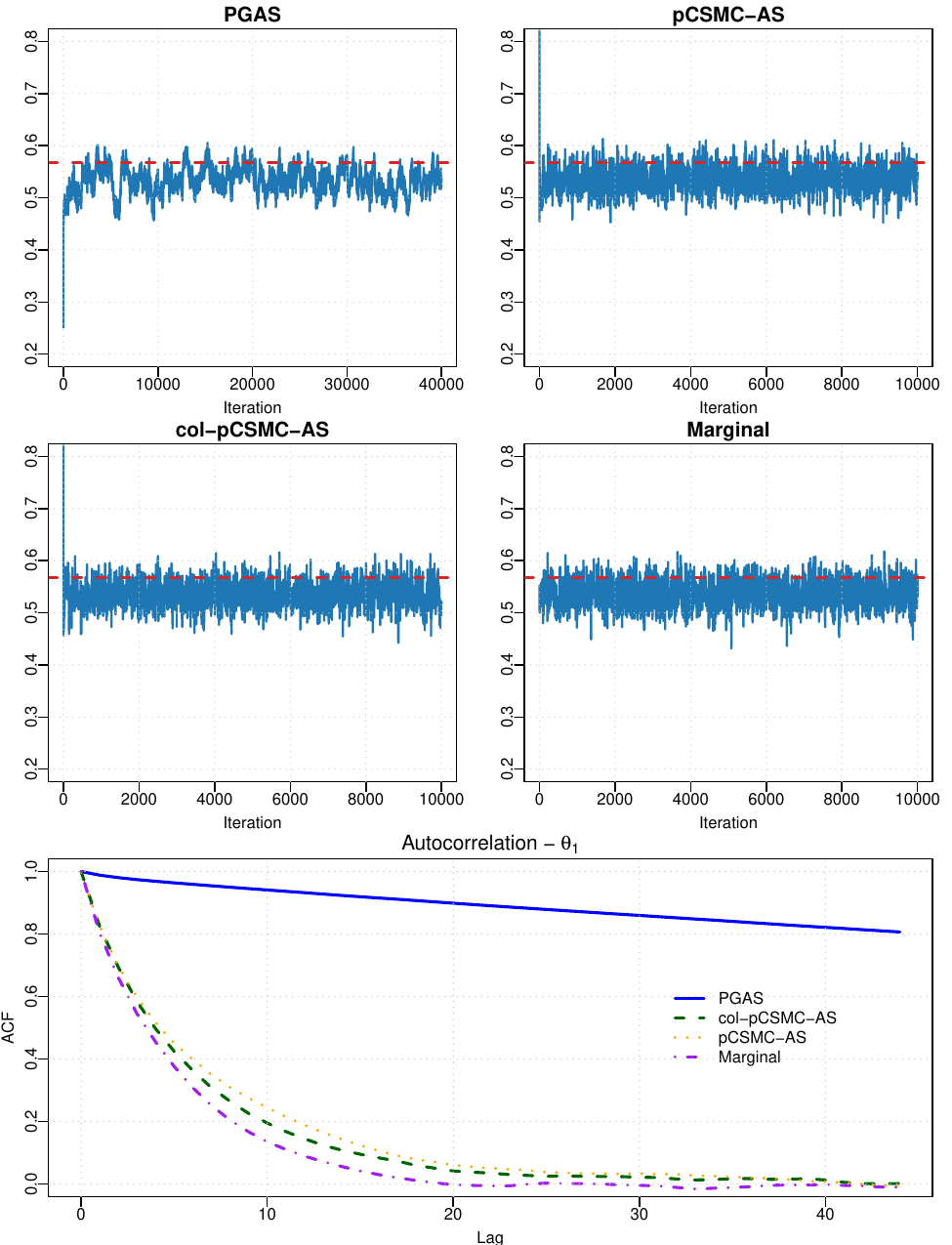}%
  \caption{Trace and autocorrelation plots of $\theta_1$ in the test with $p=2$. Burn-in periods are 4000 for the PGAS and 500 for the rest. Real values are included in the trace plots as red dashed lines for reference. Parameter samples for the marginal model were obtained by sampling from the conditional distribution available after each iteration.}
  \label{supl:fig:acf_theta_d2}
\end{figure}
\newpage
\begin{figure}[H]
  \centering
    \includegraphics[width=0.9\linewidth]{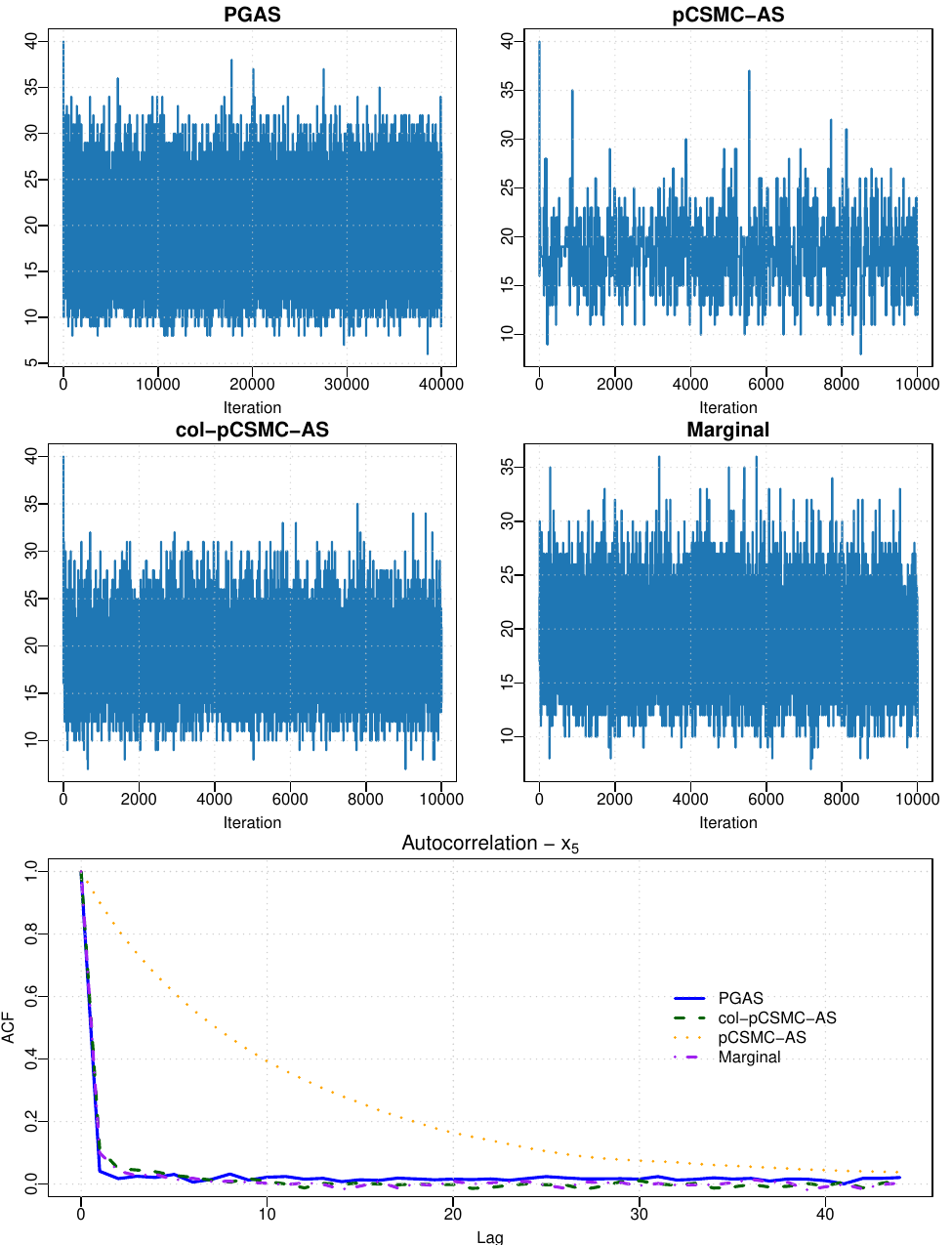}%
  \caption{Trace and autocorrelation plots of $x_5=\sum_{m=1}^px_{5,m}$ in the test with $p=2$.}
  \label{supl:fig:acf_latent_d2}
\end{figure}
\newpage
\begin{figure}[H]
  \centering
    \includegraphics[width=0.45\linewidth]{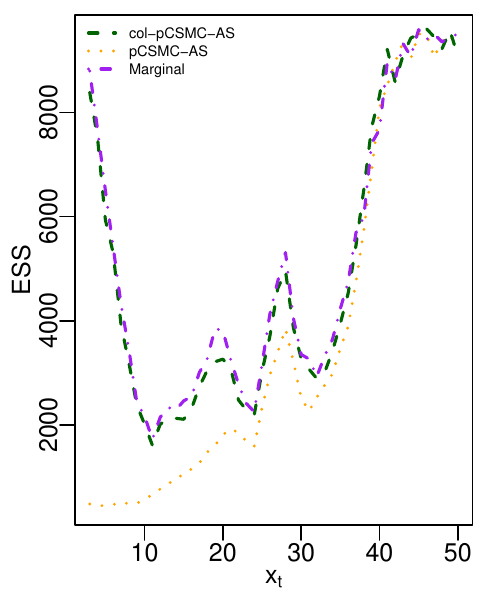}%
    \includegraphics[width=0.45\linewidth]{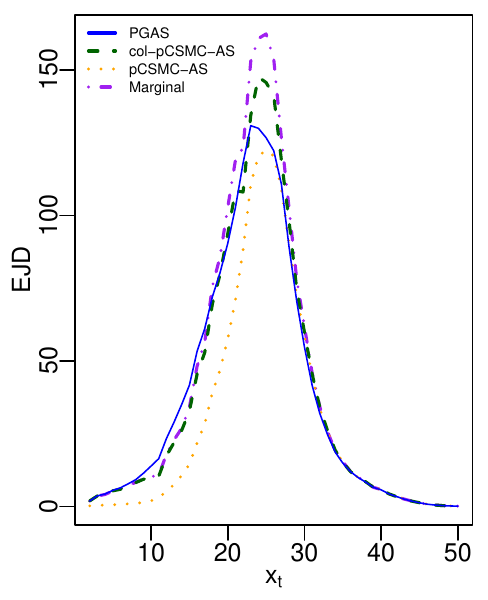}%
  \caption{Effective Sample Size (ESS) and Expected Jumping Distance (EJD) of the marginal distributions of $x_t=\sum_{m=1}^px_{t,m}$ in the test with $p=2$.}
  \label{supl:fig:ess_d2}
\end{figure}

\newpage
\subsection{p=7}

\begin{figure}[ht]
  \centering
    \includegraphics[width=0.8\linewidth]{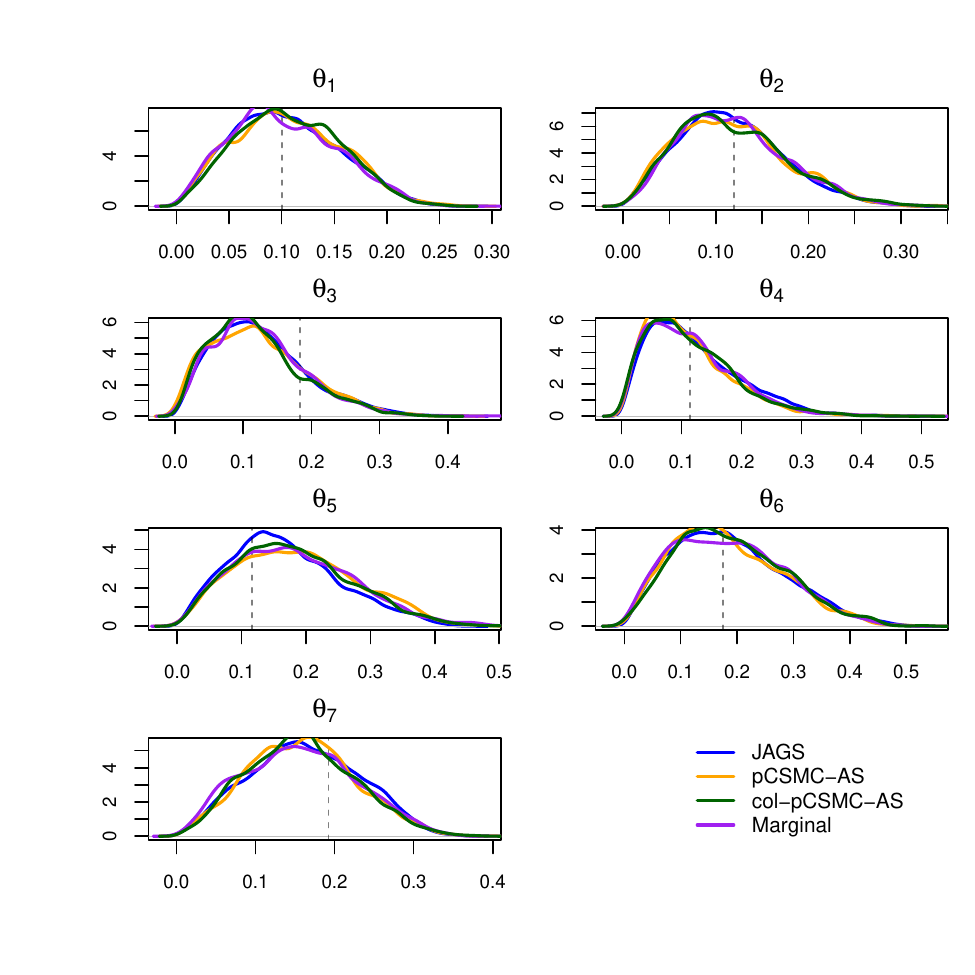}%
  \caption{Comparison of the posterior samples of the marginal distributions of the parameters for p=7 obtained with pCSMC-AS, col-pCSMC-AS and the full marginal CSMC against JAGS (slice sampler). The real values used in the generative model are shown as vertical dotted lines for reference. Parameter sampler for the full marginal model were obtained by sampling from the conditional distribution available after each iteration.}
\label{supl:fig:posteriors_d7}
\end{figure}

\newpage
\begin{figure}[H]
  \centering
    \includegraphics[width=0.8\linewidth]{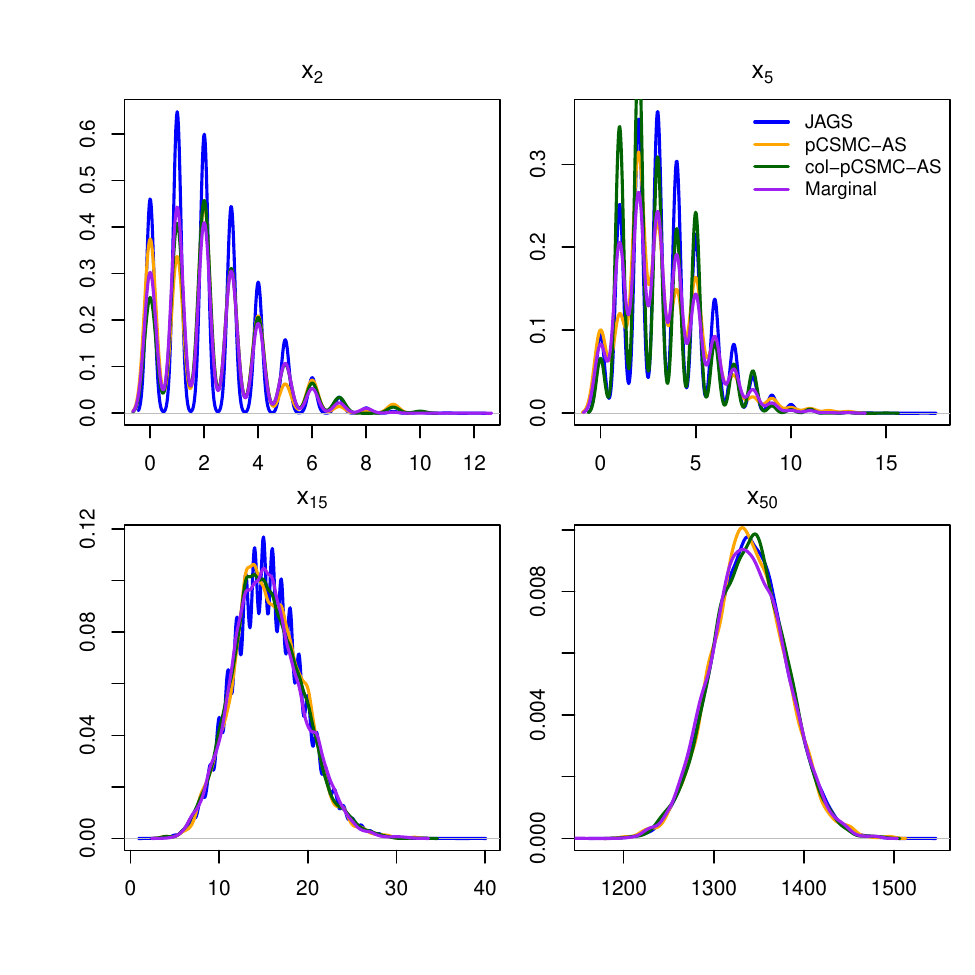}%
  \caption{Comparison of the marginal posterior distributions of $x_t=\sum_{m=1}^px_{t,m}$ for $t=\{2,5,15,50\}$ obtained with pCSMC-AS, col-pCSMC-AS and the full marginal CSMC against JAGS with $p=7$}
  \label{supl:fig:posteriors_latentd7}
\end{figure}
\newpage
\begin{figure}[H]
  \centering
    \includegraphics[width=0.9\linewidth]{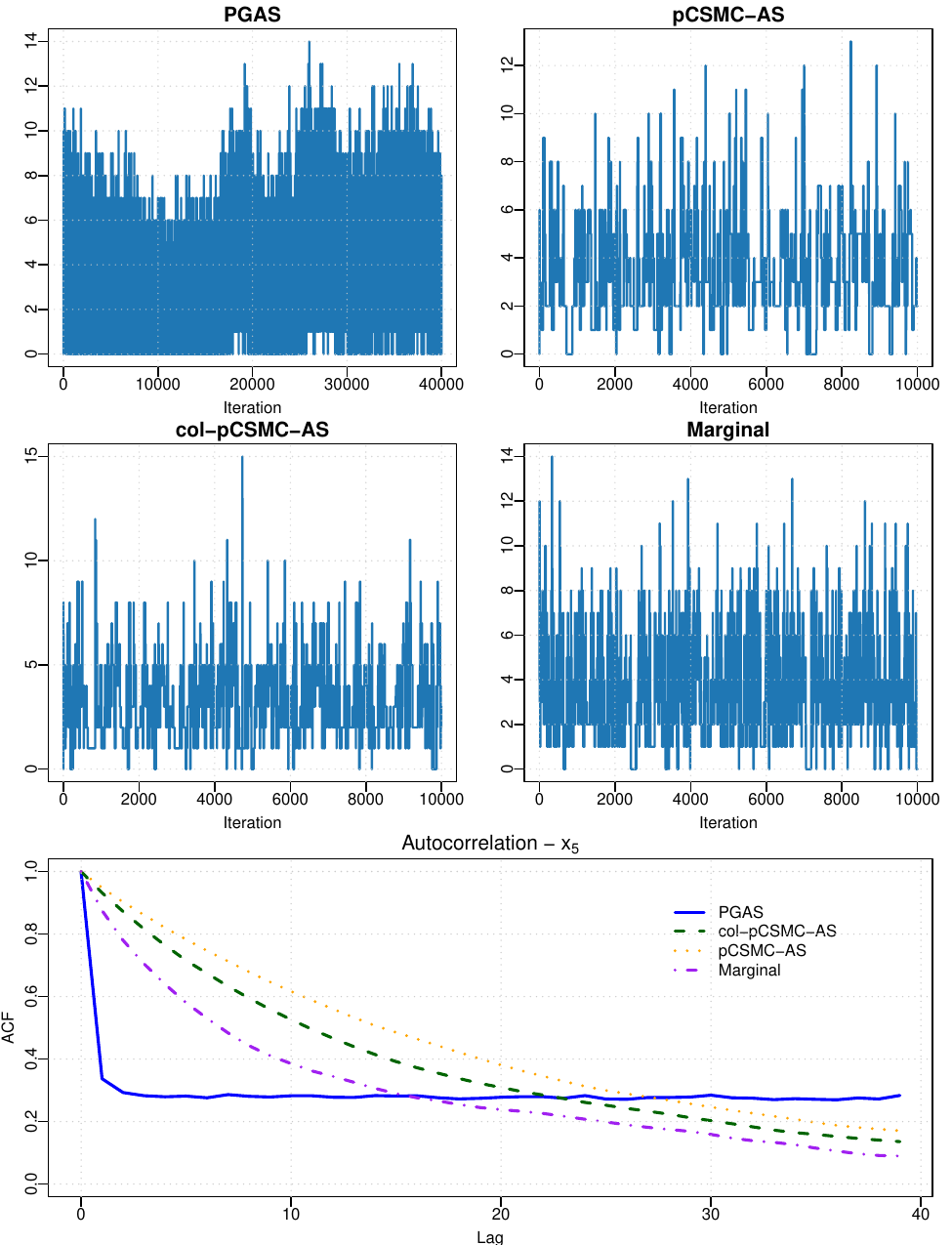}%
  \caption{Trace and autocorrelation plots of $s_5=\sum_{m=1}^px_{5,m}$ in the test with $p=7$. Burn-in periods used in the autocorrelation plots are 4000 for the PGAS and 500 for the rest. Note that PGAS did not achieve full convergence in this experiment and is only shown for comparison.}
  \label{supl:fig:acf_latent_d7}
\end{figure}
\newpage
\begin{figure}[H]
  \centering
    \includegraphics[width=0.45\linewidth]{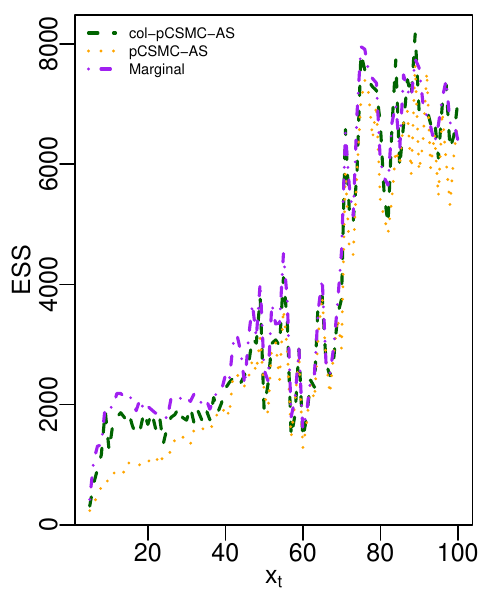}%
     \includegraphics[width=0.45\linewidth]{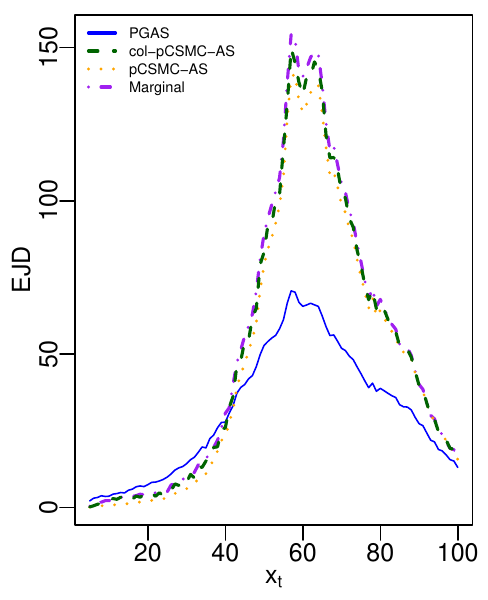}%
  \caption{Effective Sample Size (ESS) and Expected Jumping Distance (EJD) of the marginal distributions of $x_t=\sum_{m=1}^px_{t,m}$ in the test with $p=7$.}
  \label{supl:fig:ess_d7}
\end{figure}

\newpage
\section{Tests with synthetic data: Combined estimation of $\theta$ and $R_t$}\label{supl:sec:combined}

Figure \ref{supl:fig:syn_combined} shows the synthetic data generated to address combined inference of both the weights $\theta$ and the scaling parameter $R_{1:T}$ where we set p=4. 

\begin{figure}[H]
  \centering
    \includegraphics[width=0.85\linewidth]{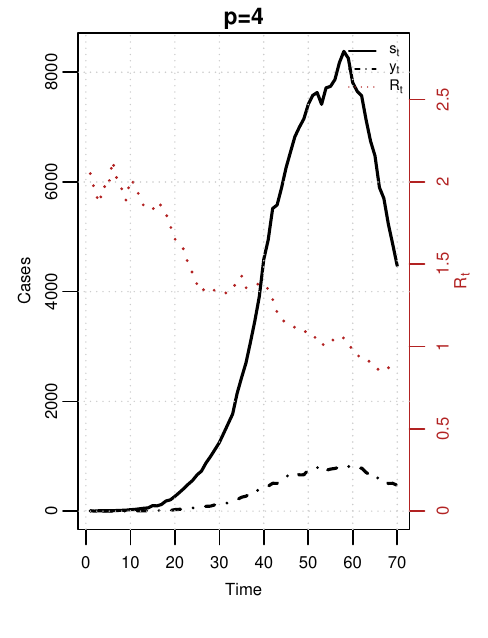}
  \caption{Synthetic data generated to address combined inference of $R_t$ and $\theta$.}
  \label{supl:fig:syn_combined}
\end{figure}

The combined inference was addressed by sequentially targetting each of the conditionals $p(\theta,\bm x_{1:T}|\bm y_{1:T},\bm R_{1:T})$ and $p(\bm R_{1:T}|\bm y_{1:T},\theta,\bm x_{1:T})$.

For the first step we have used the exact same partially collapsed implementation introduced in the previous example where $R$ was assumed known. The only difference was that, in order to avoid non-identifiability issues (there are infinite number of combinations of unconstrained $\theta$ and $R_t$ that would lead to the exact same poisson mean at each time point) we set $\theta_1=1$ so that the scale is fixed and the entire problem becomes identifiable, while holding conjugacy. For the other weights, we apply Gamma priors with $\alpha_0=2$ and $\beta_0=4$.

For the estimation of $\bm R_{1:T}$ we implemented a pure CSMC-AS algorithm to target the conditional  $p(\bm R_{1:T}|\bm y_{1:T},\theta,\bm x_{1:T})$ where we assumed autoregressive model AR(1) in the log scale with some known parameter $\phi$ and variance $\sigma^2_R$ through the latent space $R_{1:T}$:
\[
\log R_t \sim \mathcal{N}(\phi\log R_{t-1}, \sigma_R^2)
\]
Where we set $\sigma_R=0.25$ and $\phi=0.95$. 
Note that in this step, $x$ plays the role of the observations given the conditional independence $p(\bm R_{1:T}|\bm x_{1:T},\bm y_{1:T},\theta)=p(\bm R_{1:T}|\bm x_{1:T},\theta)$. The weights, for the bootstrap setting, can be then computed through Poisson likelihoods: $w_t^i=p(R^i_t|x_t,\theta)$.

The ancestor sampling weigths in this case corresponded to:
\[
\tilde w_{t-1|T}^i \propto w_{t-1}^i p(\log R^*_{t}|\log R_{t-1})
\]
for some reference path $\bm R^*_{1:T}$.

Because achieving full convergence in this setting with PGAS was completely unfeasible in useful time, we only compare the performance of our partially collapsed p-CSMC-AS algorithm against JAGS (slice sampler).

\begin{figure}[H]
  \centering
    \includegraphics[width=0.85\linewidth]{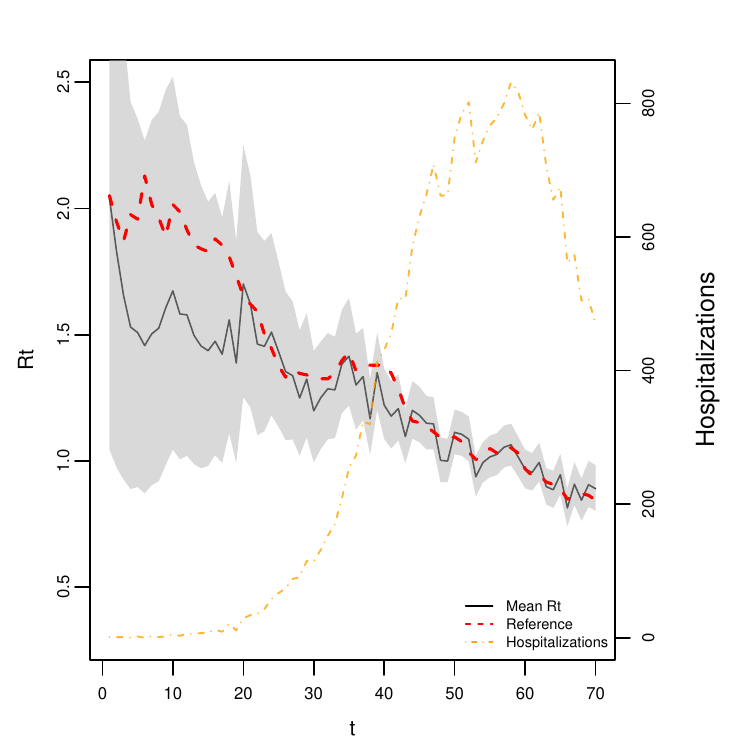}
  \caption{Posterior distribution of $R_{1:T}$ (normalized as $R_t \sum_{m=1}^p\theta_m$) obtained with the col-pCSMC-AS algorithm vs the ground truth, shown alongside the (synthetic) hospitalizations. The shaded region represents the 95\% credible interval.}
  \label{supl:fig:Rt_synthetic}
\end{figure}

\begin{figure}[H]
  \centering
    \includegraphics[width=0.85\linewidth]{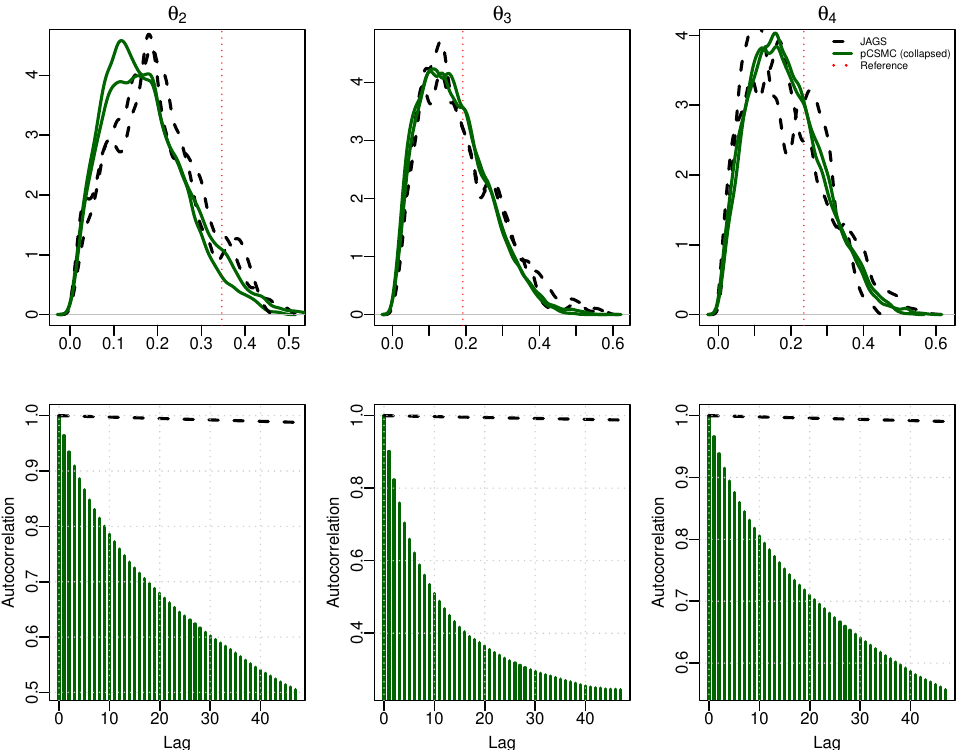}
  \caption{Comparison of posterior densities and ACF obtained through both col-pCSMC-AS and JAGS for the parameters $\theta$. Vertical dotted lines represent the ground truth values for reference.}
  \label{supl:fig:acf_combined}
\end{figure}

As figure \ref{supl:fig:acf_combined} shows, the mixing with the col-pCSMC-AS algorithm was much better than with the slice sampler. Note that, despite of the fact that 320.000 iterations were run in JAGS (120.000 burnin plus 200.000 samples) only weak convergence was obtained (i.e. Gelman-Rubin over 1.10 for some of the latent variables). On the other hand, after 60.000 iterations of the col-pCSMC-AS algorithm (including 5000 samples as burnin) full convergence was achieved.

\begin{figure}[H]
 \centering
    \includegraphics[width=0.75\linewidth]{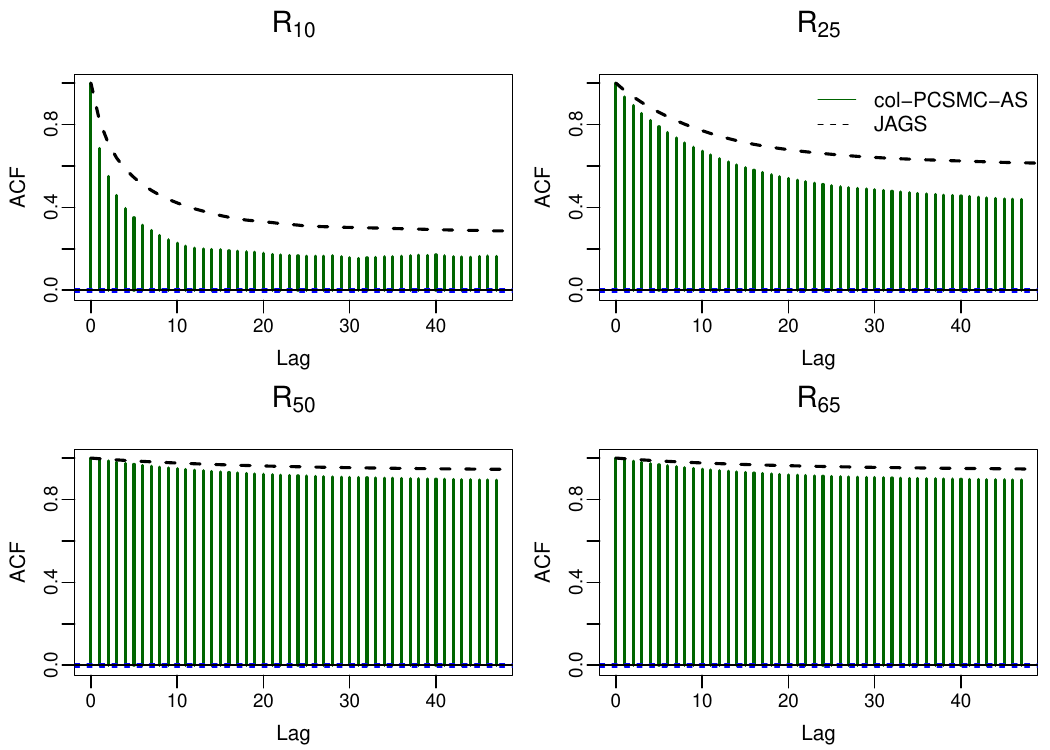}%
  \caption{Autocorrelation function of $R_t$ at different time points for the col-PCSMC-AS implementation and JAGS}
  \label{supl:fig:R_ACF}
\end{figure}

\end{document}